\documentclass[twocolumn,aps,pra,longbibliography,superscriptaddress,nofootinbib,10pt]{revtex4-2}
\usepackage[english]{babel} 	
\usepackage[usenames, dvipsnames]{color} %color management
\usepackage{graphicx} % enhanced support for graphics
\usepackage{bm} % Bold mathsymbols
\usepackage{amsmath} % improvements for mathformulas
\usepackage{amsthm} % improvements theorem-like enviouronment
\usepackage{amssymb} % extended symbols
\usepackage{times} % fontes adobe
\usepackage[pdftex,colorlinks=true,urlcolor= blue,linkcolor=Blue,citecolor=RedViolet]{hyperref}
\usepackage{multirow}
\usepackage{graphicx}
\usepackage{bm}% bold math
\usepackage{physics}
\usepackage{appendix} %for appendix

\usepackage{amssymb} % Include the amssymb package for mathematical symbols
\usepackage{graphicx} % Required for inserting images
\usepackage{amsmath} % Include the amsmath package for math environments and \text command
\usepackage{amsthm} % Include the amsthm package for theorem environments

\newtheorem{definition}{Definition}

\newtheorem{theorem}{Theorem}
\newtheorem{lemma}{Lemma}

\usepackage{dsfont}
\usepackage[small,bf]{caption}

\usepackage{soul}

\begin{document}

\title{Quantum Speed Limits Based on Schatten Norms: Universality and Tightness}

\author{Alberto J. B. Rosal}
\affiliation{Instituto de F\'{i}sica de S\~{a}o Carlos, Universidade de S\~{a}o Paulo, CP 369, 13560-970 S\~{a}o Carlos, Brazil}
\author{Diogo O. Soares-Pinto}
\affiliation{Instituto de F\'{i}sica de S\~{a}o Carlos, Universidade de S\~{a}o Paulo, CP 369, 13560-970 S\~{a}o Carlos, Brazil}
\author{Diego Paiva Pires}
\affiliation{Departamento de F\'{i}sica, Universidade Federal do Maranh\~{a}o, Campus Universit\'{a}rio do Bacanga, 65080-805, S\~{a}o Lu\'{i}s, Maranh\~{a}o, Brazil}

\begin{abstract}
We discuss quantum speed limits (QSLs) for finite-dimensional quantum systems undergoing general physical processes. These QSLs were obtained using Schatten $\alpha$-norms, firstly exploiting geometric features of the quantum state space, and secondly by applying the Holder's inequality for matrix norms. For single-qubit states, we find that the geometric QSL is independent of the chosen Schatten norm, thus revealing universality behavior. We compare these QSLs with existing speed limits in literature, showing that the latter results represent particular cases of a general class of QSLs related to Schatten $\alpha$-norms. We address necessary and sufficient conditions for the tightness of the QSLs that depends on populations and coherences of the qubits, also addressing their geometric meaning. We compare the QSLs obtained for qubit dynamics, also exploring their geometrical meaning. Finally, we show that the geometric QSL is tighter for general qubit dynamics with initial pure states, which indicates a universal QSL.
\end{abstract}

\maketitle

\section{Introduction}
\label{sec:000000001}

Quantum speed limits (QSLs) are of fundamental interest in several branches of physics, ranging from non-equilibrium thermodynamics~\cite{DL2010,arXiv:2204.10368,arXiv:2006.14523,arXiv:2203.12421}, to quantum many-body systems~\cite{PhysRevLett.124.110601,PhysRevLett.126.180603,PhysRevResearch.2.032020,PhysRevA.104.052223,PhysRevX.12.011038}, also including quantum metrology~\cite{Giovannetti}, communication and quantum computing~\cite{JBekenstein,Loyd}, quantum entanglement~\cite{EPL_62_615_2003,PhysRevA.72.032337,PhysRevA.73.049904,arXiv:2401.04599}, quantum-to-classical transition~\cite{PhysRevLett.120.070401,PhysRevLett.120.070402,PRXQuantum.3.020319}, and optimal control theory~\cite{Caneva,PhysRevA.102.042606}. Furthermore, QSLs also play a role in the study of quantum batteries~\cite{PhysRevA.106.042436,PhysRevResearch.2.023113,arXiv:2308.16086}, and also non-Hermitian systems~\cite{PhysRevA.104.052620,ChinPhysB_29_030304,PhysRevA.106.012403,JPhysA_57_025302_2023,arXiv:2404.16392}. Interestingly, recent works addressed exact speed limits for the unitary dynamics of quantum systems that outperform previous QSLs~\cite{arXiv:2305.03839}, also discussing exact and inexact speed limits for completely positive and trace preserving dynamics for finite-dimensional quantum systems in Liouville space~\cite{arXiv:2406.08584}. Recent studies include experimental investigations of QSLs in both platforms of trapped atoms~\cite{sciadvabj91192021}, and Nuclear Magnetic Resonance systems~\cite{PhysRevA.97.052125,arXiv:2307.06558}. In short, the QSL sets the shortest time in which a quantum system can evolve from some initial to final state under a given dynamics. We refer to Refs.~\cite{Deffner_2017,IntJPhysB362230007} for recent reviews covering the subject.

For closed quantum systems, Mandelstam and Tamm (MT)~\cite{Mandelstam1991} proved that $\tau \geq {\tau_{MT}}$, with ${\tau_{MT}} := ({\hbar}/{\Delta E})\arccos{(|\langle{\psi_0}|{\psi_{\tau}}\rangle|)}$, while $\Delta E = \sqrt{\langle{H^2}\rangle - {\langle{H}\rangle^2}}$ is the variance of the time-independent Hamiltonian $H$ that generates the dynamics, whereas $|{\psi_{0}}\rangle$ and $|{\psi_{\tau}}\rangle$ set the initial and final pure states of the system, respectively. Particularly, for orthogonal states, one obtains ${\tau_{MT}} = \pi{\hbar}/(2{\Delta E})$. In turn, Margolus and Le\-vi\-tin~\cite{1992_PhysicaD_120_188} realized the bound $\tau \geq {\tau_{ML}}$ for unitary dynamics and maximally distinguishable states, with ${\tau_{ML}} := \hbar\pi/[2(\langle{\psi_0}|{H}|{\psi_0}\rangle - {E_0})]$, where $E_0$ is the ground state of $H$. We shall mention another QSL bound for closed systems with bounded energy spectrum that is dual to the ML bound~\cite{PhysRevLett.129.140403}. Noteworthy, Levitin and Toffolli (LT) showed that the tighter QSL for orthogonal states is achieved by combining MT and ML bounds as $\tau_{QSL} = \text{max}\{ {\tau_{MT}}, {\tau_{ML}}\}$~\cite{PhysRevLett.103.160502}. In addition, QSLs have also been addressed for driven closed quantum systems and mixed states~\cite{PhysRevA.67.052109,PhysRevA.82.022107,PhysRevA.86.016101,JPhysAMathTheor_46_335302,PhysLettA380_689,PhysRevE.103.032105,Kavan}. Recent investigations include the derivation of extended MT and ML bounds for states with arbitrary fidelity~\cite{NJPhys240550042022,arXiv:2301.10063}. We also mention speed limits for non-Hermitian systems that are related to geometric phases of the dynamics, which in turn can be tighter than Mandelstam-Tamm and Margolus-Levitin bounds in some cases~\cite{PhysRevLett.123.180403}.

QSLs have been also investigated for open quantum systems. Taddei {\it et al}.~\cite{Taddei2013} and del Campo {\it et al}.~\cite{delCampoQSL} originally addressed such issue for nonunitary evolutions, thus deriving MT-like bounds for general physical processes. The former bound is based on the quantum Fisher information, while the latter is obtained from the relative purity between initial and final states. In turn, Deffner and Lutz~\cite{DLQSL} proposed a QSL time that is based on operator norms and the Bures angle, also indicating that non-Markovianity can speed up the evolution. However, it has been shown that the QSL time is sensitive to the initial state and details of the dynamical evolution~\cite{Sun2015,Meng2015,Nicolas2016,Cianciaruso2017,Teittinen2019,Entropy_23_331_2021}. QSLs have been addressed within the perspective of information geometry~\cite{PhysRevLett.65.1697,Ingemar_Bengtsson_Zyczkowski}. This approach provides a family of QSLs based on contractive Riemannian metrics equipping the space of quantum states, valid for unitary and nonunitary dynamics, mixed and pure states~\cite{DDQSL,PhysRevA.103.022210,NJPhys_24_055003_2022}. Recently, the geome\-tric approach has been used to obtain tight and attainable QSLS for nonunitary dynamics~\cite{PhysRevA.108.052207}. We also mention the interplay between QSLs and Finsler metrics, which finds applications in quantum information processing~\cite{10.1142_S0129054114400073}.

QSLs have also been addressed by using families of matrix norms. For example, Ref.~\cite{NewJPhys_19_103018_2017} discusses a QSL based on Schatten $\alpha$-norms obtained by using an approach that is restricted to the Wigner phase space. In turn, Ref.~\cite{PhysRevA.97.022109} comment on the interplay between QSLs and statistical speeds for unitary dynamics with local Hamiltonians. In addition, Refs.~\cite{NewJPhys_21_013006_2019,PhysRevLett.126.010601,NewJPhys_24_095004_2022} study speed limits for open quantum systems based on trace distance. Recently, Ref.~\cite{arXiv:2401.01746} introduced a family of coherent QSL bounds based on Schatten $\alpha$-norms that only apply to closed quantum systems. Based on the different QSLs present in the literature, each related to some measure of distinguishability of quantum states, would it be possible to generalize them to obtain an ultimate QSL that would be valid for any Schatten norm and general physical processes? If so, how could this be? How does a given QSL change when choosing another information metric to equip the space of quantum states? For example, to the best of our knowledge, there is no clear relationship in literature between the QSL obtained using the Bures angle~\cite{DLQSL} with those QSLs obtained using the Hilbert-Schmidt norm~\cite{Kavan}. 

We will seek to answer these questions in this work, bringing an idea of universality between different types of QSL. In this work, we aim to develop families of QSLs based on general Schatten $\alpha$-norms defined in the space of quantum states that holds for arbitrary physical processes, thereby generalizing other types of QSLs that already exist in the literature. We also compare the QSL based on Schatten norms and the QSL developed from the geometrical approach with those existing in the literature~\cite{DLQSL,delCampoQSL,Kavan}. Furthermore, we analyze the necessary and sufficient conditions under the dynamics for the QSL induced by Schatten norms to be tight. To clarify, given an evolution from an initial state ${\rho_0}$ to a final one ${\rho_{\tau}}$, described by a dynamical map $\Lambda_t$ and a QSL quantifier $\tau_{QSL}$, it can be said tight if we have that $\tau = \tau_{QSL}$. On the contrary, if we have two QSL quantifiers $\tau_{QSL}^{(1)}$ and $\tau_{QSL}^{(2)}$, we set that $\tau_{QSL}^{(1)}$ is \textit{tigher} than $\tau_{QSL}^{(2)}$ if we have that $\tau_{QSL}^{(1)} \geq \tau_{QSL}^{(2)}$.

This paper is organized as follows. In Sec.~\ref{sec:000000002}, we develop a general QSL induced from a Schatten norm through the geometrical approach. In Sec.~\ref{sec:000000003}, we address another family of QSLs by exploiting the Holder's inequality. We discuss the relation between these QSLs with remarkable results previously derived in literature, also presenting the necessary and sufficient condition for the tightness of the QSLs. In Sec.~\ref{sec:000000004}, we compare these two families of QSLs and derive an upper bound between them for the case of single-qubit states. In Sec.~\ref{sec:000000005}, we present our conclusions. Finally, technical details are organized in the Appendices.

%%%%%%%%%%%%%%%%%%%%%%%%%%%%%%%%%
%%%%%%%%%%%%%%%%%%%%%%%%%%%%%%%%%
%%%%%%%%%%%%%%%%%%%%%%%%%%%%%%%%%
%%%%%%%%%%%%%%%%%%%%%%%%%%%%%%%%%

\section{Geometric QSL from normed space of quantum states}
\label{sec:000000002}

%%%%%%%%%%%%%%%%%%%%%%%%%%%%%%%%%
%%%%%%%%%%%%%%%%%%%%%%%%%%%%%%%%%
%%%%%%%%%%%%%%%%%%%%%%%%%%%%%%%%%
%%%%%%%%%%%%%%%%%%%%%%%%%%%%%%%%%

\subsection{$\alpha$-QSL}
\label{sec:000000002A}

We consider a quantum system related to a finite-dimensional Hilbert space $\mathcal{H}$, with $d = \text{dim}(\mathcal{H})$, while $\mathcal{D}(\mathcal{H}) = \{ \rho \in \mathcal{H} \, | \, {\rho^{\dagger}} = \rho, \rho \geq 0, \text{Tr}({\rho}) = 1\}$ defines the convex set of quantum states. In turn, the space of quantum states is equipped with the family of matrix norms $\|\bullet\|$, which is unitarily invariant, i.e., $\|V\bullet{V^{\dagger}}\| = \|\bullet\|$ for any unitary operator ${V^{\dagger}} = {V^{-1}}$ such that $V{V^{\dagger}} = {V^{\dagger}}V = \mathbb{I}$. In this setting, the function $D(\rho,\varrho) = \|\rho - \varrho\|$ defines a \textit{bona fide} distance over the space of quantum states, for all $\rho,\varrho \in \mathcal{D}(\mathcal{H})$. It can be proved that the pair $(\mathcal{D}(\mathcal{H}),D(\rho,\varrho))$ is a metric space [see Appendix \ref{app:00000000A}]. In this case, the \textit{geodesic path} between the two states $\rho,\varrho \in D(\rho,\varrho)$ is the path that minimize the distance between them.

Let $\gamma \equiv \gamma(\rho,\varrho) \subset \mathcal{D}(\mathcal{H})$ be a generic curve that connects the two quantum states $\rho,\varrho \in \mathcal{D}(\mathcal{H})$ in the space of quantum states. Then, the length of this curve is ${L_{\gamma}} := {\int_{\gamma}} \sqrt{ds^2}$, where $ds$ is the infinitesimal distance between neighboring states states $\rho$ and $\rho + d\rho$ as follows
\begin{equation}
\label{eq:0000000000001}
ds := D(\rho + d\rho,\rho) = \|d\rho\| ~.
\end{equation}
Next, by using the triangular inequality satisfied by the norms, one gets
\begin{equation}
\label{eq:0000000000002}
{L_{\gamma}}  = {\int_{\gamma}} \|d\rho\| \geq \left\| {\int_{\gamma}} d\rho \right\| = \|\rho - \varrho\| = D(\rho,\varrho) ~.
\end{equation}
Therefore, for every curve $\gamma \equiv \gamma(\rho,\varrho)$ connecting two points $\rho$ and $\varrho$ in the space of quantum states, Eq.~\eqref{eq:0000000000002} shows that the distance $D(\rho,\varrho) = \|\rho - \varrho \|$ sets a lower bound to the length ${L_{\gamma}}$ of such curve. Symbolically, we can write
\begin{equation}
\label{eq:0000000000003}
\forall \gamma \equiv \gamma(\rho,\varrho): ~{L_{\gamma}} \geq D(\rho,\varrho) ~.
\end{equation}
The path that saturate the inequality in Eqs.~\eqref{eq:0000000000002} and~\eqref{eq:0000000000003} is given by the straight line connecting the fixed states $\rho,\varrho \in \mathcal{D}(\mathcal{H})$, given by the convex parameterization ${\rho_{\lambda}} := (1 - \lambda)\rho + \lambda\varrho$, with $\lambda \in [0,1]$ being a continuous, real-valued, and positive parameter. We note that ${\rho_{\lambda}^{\dagger}} = {\rho_{\lambda}}$, ${\rho_{\lambda}} \geq 0$, $\text{Tr}({\rho_{\lambda}}) = 1$, with ${\rho_0} = \rho$ and ${\rho_1} = \varrho$. In this case, one finds the infinitesimal distance $ds = D({\rho_{\lambda + d\lambda}},{\rho_{\lambda}}) =  \|\rho - \varrho\| d\lambda$, which implies that 
\begin{equation}
\label{eq:0000000000004}
{L_{\gamma_{\text{line}}}} = {\int_0^1} \,d\lambda\, \|\rho - \varrho\| = \|\rho - \varrho\| = D(\rho,\varrho) ~. 
\end{equation}
Equation~\eqref{eq:0000000000004} shows that the minimal length is obtained when we have a linear path connecting two points $\rho,\varrho\in\mathcal{D}(\mathcal{H})$ in the space of quantum states. Hence, one finds that the linear path is the geodesic path.

Let $\gamma_{\tau} = \{ {\rho_t} = {\Lambda_t}({\rho_0}): t \in [0,\tau] \}$ be the path related to the dynamical map ${\Lambda_t}(\bullet)$, where the system evolves from the initial state ${\rho_0}$ to the final state ${\rho_{\tau}} = {\Lambda_{\tau}}({\rho_0})$. In this case, we have a point $\rho_t$ in $\mathcal{D}(\mathcal{H})$ for all $0 \leq t \leq \tau$, and one obtains a continuous path given by the curve $\gamma_{\tau}$. In this setting, the length $L_{\gamma_{\tau}}$ of the path $\gamma_{\tau}$ that connects the quantum states ${\rho_0}, {\rho_{\tau}} \in \mathcal{D}(\mathcal{H})$ reads
\begin{equation}
\label{eq:0000000000005}
{L_{\gamma_{\tau}}} = {\int_{\gamma_{\tau}}} \|d{\rho_t}\| = {\int_0^{\tau}} dt \left\| \frac{d{\rho_t}}{dt} \right\| ~,
\end{equation}
In Eq.~\eqref{eq:0000000000005}, we note that $\left\| {d{\rho_t}}/{dt} \right\|$ plays the role of a genera\-li\-zed statistical speed in the evolution of the quantum system, where ${d{\rho_t}}/{dt}$ describes the dynamical evolution of the system~\cite{PhysRevA.97.022109}. Considering that we have a dynamical evolution described by a smooth path in the state space, Eq.~\eqref{eq:0000000000005} implies that [see also Eq.~\eqref{eq:0000000000002}]
\begin{equation}
\label{eq:0000000000006}
{\int_{0}^{\tau}} dt \left\| \frac{d{\rho_t}}{dt} \right\| \geq \| {\rho_{\tau}} - {\rho_0} \| ~.
\end{equation}
From Eq.~\eqref{eq:0000000000006}, we obtain a lower bound on the evolution time given by
\begin{equation}
\label{eq:0000000000007}
\tau \geq {\tau_{QSL}^{\|\bullet\|}} ~,
\end{equation}
where ${\tau_{QSL}^{\|\bullet\|}} \equiv \tau_{QSL}^{\|\bullet\|}(\tau)$ sets the geometric QSL time induced from the norm $\|\bullet\|$ that is defined as follows
\begin{equation}
\label{eq:0000000000008}
\tau_{QSL}^{\|\bullet\|} := \frac{ \| {\rho_{\tau}} - {\rho_0} \|}{ {\int_0^{\tau}} dt \left\| \frac{d{\rho_t}}{dt} \right\| }\tau ~.
\end{equation}
We note that Eqs.~\eqref{eq:0000000000006},~\eqref{eq:0000000000007} and~\eqref{eq:0000000000008} hold for general physical processes and applies for any family of matrix norms, thus generalizing the results discussed in Refs.~\cite{NewJPhys_19_103018_2017,PhysRevA.97.022109,NewJPhys_21_013006_2019,PhysRevLett.126.010601,NewJPhys_24_095004_2022,arXiv:2401.01746}. We also mention the geometric bounds in Refs.~\cite{PhysRevA.82.022107,DDQSL,PhysRevLett.123.180403}. In detail, Ref.~\cite{PhysRevA.82.022107} addresses a Mandelstam-Tamm inequality for unitary dynamics by means of the Fisher information, also deriving a Margolous-Levitin inequality that is found to be incorrect~\cite{PhysRevA.86.016101,JPhysAMathTheor_46_335302}. Importantly, Ref.~\cite{JPhysA_51_318001_2023} showed that the derivation of the ML bound for unitary dynamics in Ref.~\cite{JPhysAMathTheor_46_335302} does not apply to time-dependent general Hamiltonians. Recently, Ref.~\cite{PhysRevA.108.052421} pointed out the non-existence of ML bounds for systems with time-dependent Hamiltonians, except for adiabatic evolutions~\cite{JPhysA_37_L157_2004}. Next, Ref.~\cite{PhysRevLett.123.180403} derived a QSL bound for non-Hermitian systems by exploiting geometric phases via the gauge-invariant Fubini-Study metric. Furthermore, the family of speed limits in Eq.~\eqref{eq:0000000000008} is substantially distinct from the QSLs addressed in Ref.~\cite{DDQSL}. Indeed, the latter present a geometric framework to address QSLs for general physical processes based on a family of contractive Riemannian metrics that equip the space of quantum states. This family of information metrics is predicted by the so-called Morozova-Cencov-Petz (MCP) theorem, which in turn encompasses Riemannian metrics related to quantum Fisher information, and also Wigner-Yanase skew information, to name a few. To the best of our knowledge, the inner product encoded by the MCP theorem does not cover metric-based matrix norms such as $\alpha$-Schatten norms.

Hereafter, we set the matrix norm $\|\bullet\|$ as the Schatten $\alpha$-norm ${\|\bullet\|_{\alpha}}$, with $\alpha \in [1,\infty)$. In this setting, the metric space $(\mathcal{D}(\mathcal{H}),D(\rho,\varrho))$ is equipped with the distance ${D_{\alpha}}(\rho,\varrho) = {\| \rho - \varrho \|_{\alpha}}$, where $\| A \|_{\alpha} := (\mathrm{Tr}\{(\sqrt{A^\dagger A})^{\alpha}\})^{\frac{1}{\alpha}}$ [see Appendix~\ref{app:00000000A}]. Hence, from Eq.~\eqref{eq:0000000000008}, the $\alpha$-QSL induced from this norm is given by
\begin{equation}
\label{eq:0000000000009}
{\tau_{QSL}^{\alpha}} (\tau) := \frac{ \| {\rho_{\tau}} - {\rho_0} \|_\alpha}{ {\int_{0}^{\tau}} dt \left\| \frac{d{\rho_t}}{dt} \right\|_\alpha} \, \tau ~.
\end{equation}

%%%%%%%%%%%%%%%%%%%%%%%%%%%%%%%%%
%%%%%%%%%%%%%%%%%%%%%%%%%%%%%%%%%
%%%%%%%%%%%%%%%%%%%%%%%%%%%%%%%%%
%%%%%%%%%%%%%%%%%%%%%%%%%%%%%%%%%

\subsection{$\alpha$-QSL for closed and open quantum systems}
\label{sec:000000002B}

In the following, we specify the physical process in order to recast the result in Eq.~\eqref{eq:0000000000009} in terms of well-known information theoretic quantifiers. From now on, we set $\hbar = 1$. We shall begin addressing closed quantum systems. Let $H$ be the time-independent Hamiltonian that generates the dynamics of the system, with $\rho_0$ being the initial state of the finite-dimensional quantum system, where ${\rho_0^{\dagger}} = {\rho_0}$, ${\rho_0} \geq 0$, and $\text{Tr}({\rho_0}) = 1$. The instantaneous state of the system is given by ${\rho_t} = {\Lambda_t}({\rho_0}) = {U_t}{\rho_0}{U_t^{\dagger}}$, where the unitary operator ${U_t} = {e^{- i t H}}$ satisfies the relation ${U_t}{U_t^{\dagger}} = {U_t^{\dagger}}{U_t} = \mathbb{I}$. The time-dependent state ${\rho_t}$ satisfies the von Neumann equation $d{\rho_t}/dt = -i [H,{\rho_t}] = -i\,{U_t}[H,{\rho_0}]{U_t^{\dagger}}$. In this case, Eq.~\eqref{eq:0000000000009} becomes
\begin{equation}
\label{eq:0000000000010}
{\tau_{QSL}^{\alpha}} (\tau) := \frac{\|{\rho_{\tau}} - {\rho_0} \|_{\alpha}}{{\int_0^{\tau}} dt \left\| \frac{d{\rho_t}}{dt} \right\|_{\alpha} }\tau = \frac{ \| {\rho_{\tau}} - {\rho_0} \|_{\alpha}}{\left\|(-i)[H,{\rho_0}]\right\|_{\alpha}}  ~,
\end{equation}
where we used the fact that the Schatten $\alpha$-norm is unitari\-ly invariant. We note that the contribution $\left\|(-i)[H,{\rho_0}]\right\|_{\alpha}$ captures the quantum fluctuations of the Hamiltonian of the system with respect to the initial state of the system. In particular, for $\alpha = 2$, one can show that ${\left\|(-i)[H,{\rho_0}]\right\|_2^2} = 4\,{\mathcal{I}_L}({\rho_0},H)$, where ${\mathcal{I}_L}({\rho_0},H) := -(1/4)\text{Tr}({[{\rho_0},H]^2})$ defines a coherence quantifier~\cite{PhysRevLett.113.170401}. Hence, for $\alpha = 2$, Eq.~\eqref{eq:0000000000010} can be written as ${\tau_{QSL}^{\alpha = 2}} (\tau) = {\|{\rho_{\tau}} - {\rho_0}\|_2}/\sqrt{4\,{\mathcal{I}_L}({\rho_0},H)}$. For initial pure states, it follows that ${\mathcal{I}_L}({\rho_0},H) = (1/2){(\Delta{E})^2}$, where ${(\Delta{E})^2} = { {\langle{H^2}\rangle_{\rho_0}} - {\langle{H}\rangle^2_{\rho_0}}}$ is the variance of the Hamiltonian with respect to the pure state $\rho_0$, with ${\langle{H^n}\rangle_{\rho_0}} := \text{Tr}({\rho_0}{H^n})$ for all $n \in \mathbb{N}$. In this case, one gets the QSL time as ${\tau_{QSL}^{\alpha = 2}}(\tau) = {\|{\rho_{\tau}} - {\rho_0}\|_2}/(\sqrt{2}{\Delta{E}})$. For initial mixed states, it has been proved that ${\mathcal{I}_L}({\rho_0},H) \leq {\mathcal{I}_{\text{WY}}}({\rho_0},H) \leq {(\Delta{E})^2}$~\cite{PhysRevLett.113.170401,Yanagi_2010}, where ${\mathcal{I}_{\text{WY}}}({\rho_0},H) = -(1/2)\text{Tr}({[\sqrt{\rho_0},H]^2})$ is the Wigner-Yanase skew information. In this case, one finds the lower bound ${\tau_{QSL}^{\alpha = 2}}(\tau) \geq {\|{\rho_{\tau}} - {\rho_0}\|_2}/(2{\Delta{E}})$, which is also inversely proportional to the variance $\Delta{E}$. These results show that, for closed quantum systems driven by time-independent Hamiltonians, the QSL time fits into the class of QSLs \textit{\`{a} la} Mandelstam-Tamm, regardless the initial state being a pure or mixed one.

Next, for open quantum systems, let ${\rho_t} = {\Lambda_t}({\rho_0}) = {\sum_j}{K_j}{\rho_0}{K_j^{\dagger}}$ be the evolved state of the system, where $\{{K_j}\}$ depicts the set of time-dependent Kraus operators related to the nonunitary quantum operation ${\Lambda_t}(\bullet)$ that models the dynamics, with ${\sum_j}{K_j^{\dagger}}{K_j} = \mathbb{I}$. It can be shown that ${\|d{\rho_t}/dt\|_{\alpha}} \leq 2\,{\sum_j}\,{\|{K_j}{\rho_0}{\dot{K}_j^{\dagger}}\|_{\alpha}}$, where we have applied the triangle inequality ${\|A + B\|_{\alpha}} \leq {\|A\|_{\alpha}} + {\|B\|_{\alpha}}$, and also used the fact that the norm $\|\bullet\|$ is unitarily invariant such that ${\|{\dot{K}_j}{\rho_0}{K_j^{\dagger}}\|_{\alpha}} = {\|{K_j}{\rho_0}{\dot{K}_j^{\dagger}}\|_{\alpha}}$. Hence, Eq.~\eqref{eq:0000000000009} becomes
\begin{equation}
\label{eq:0000000000011}
{\tau_{QSL}^{\alpha}}(\tau) = \frac{ \| {\rho_{\tau}} - {\rho_0} \|_{\alpha}}{ {\int_0^{\tau}} dt \left\| \frac{d{\rho_t}}{dt} \right\|_{\alpha} }\tau \geq \frac{ \| {\rho_{\tau}} - {\rho_0} \|_{\alpha}}{2\,{\sum_j}\,{\int_0^{\tau}}\, dt \left\|{K_j}{\rho_0}{\dot{K}_j^{\dagger}}\right\|_{\alpha}}\tau ~.
\end{equation}

%%%%%%%%%%%%%%%%%%%%%%%%%%%%%%%%%
%%%%%%%%%%%%%%%%%%%%%%%%%%%%%%%%%
%%%%%%%%%%%%%%%%%%%%%%%%%%%%%%%%%
%%%%%%%%%%%%%%%%%%%%%%%%%%%%%%%%%

\subsection{$\alpha$-QSL for qubits}
\label{sec:000000002C}

In the following, we consider the dynamics of single-qubit states, ${\rho_t} = {\Lambda_t}({\rho_0}):= (1/2)(\mathbb{I} + {\vec{n}_t}\cdot\vec{\sigma})$, whose evolution is governed by general dynamical maps, where ${\vec{n}_t} = ({n_{x,t}},{n_{y,t}},{n_{z,t}})$ is the Bloch vector for every $t \geq 0$, and  ${\vec{\sigma}} = ({\sigma_{x}},{\sigma_{y}},{\sigma_{z}})$ are the Pauli matrices. In this case, we can compute explicitly the distance $\|  {\rho_{\tau}} - {\rho_0} \|_\alpha$ as well as the norm $\|d{\rho_t}/dt\|_{\alpha}$ [see Appendix~\ref{app:00000000B}]. We state the following result:
%%%%%%%%%%%%%%%%%%%%%%%%%%%%%%%%%%%%%%%%
%%%%%%%%%%%%%%%%%%%%%%%%%%%%%%%%%%%%%%%%
\begin{theorem}[Universality of the $\alpha$-QSL for qubits]
\label{theo:THEOREM000001}
Let's consider an initial qubit state ${\rho_0} = (1/2)(\mathbb{I} + {\vec{n}_0}\cdot\vec{\sigma})$. The dynamics of the system is given in terms of a general dynamical map ${\Lambda_{t,0}}(\bullet)$, and its instantaneous state becomes ${\rho_t} = {\Lambda_{t,0}}({\rho_0}) = (1/2)(\mathbb{I} + {\vec{n}_t}\cdot\vec{\sigma})$, for all $t \in [0,\tau]$. From Eq.~\eqref{eq:0000000000009}, the geometric $\alpha$-QSL that is induced from Schatten $\alpha$-norm, for every $\alpha \in [1, \infty)$, yields
\begin{equation}
\label{eq:0000000000012}
\tau_{QSL}^\alpha(\tau) = \frac{d_{\text{euc}}({\vec{n}_0},{\vec{n}_{\tau}})}{{\int_0^{\tau}} dt \left\| \frac{d{\vec{n}_t}}{dt} \right\|} \, \tau ~,
\end{equation}
with
\begin{equation}
\label{eq:0000000000013}
\left\| \frac{d{\vec{n}_t}}{dt} \right\| = \sqrt{{\sum_{j = x,y,z}}\,{\left(\frac{d{n_{j,t}}}{dt}\right)^2}} ~,
\end{equation}
while
\begin{equation}
\label{eq:0000000000014}
d_{\text{euc}}({\vec{n}_0}, {\vec{n}_{\tau}}) = \sqrt{ {\sum_{j = x,y,z}} ({n_{j,\tau}} - {n_{j,0}})^2}
\end{equation}
stands for the Euclidean distance between the initial and the final Bloch vectors in the Bloch space. 
\end{theorem}
%%%%%%%%%%%%%%%%%%%%%%%%%%%%%%%%%%%%%%%%
%%%%%%%%%%%%%%%%%%%%%%%%%%%%%%%%%%%%%%%%
\begin{proof}
See Appendix~\ref{app:00000000C}.
\end{proof}

Theorem~\ref{theo:THEOREM000001} shows that there is a universality of the geometric QSL induced through of a $\alpha$-norm. That is, no matter which Schatten $\alpha$-norm is chosen, Eq.~\eqref{eq:0000000000012} gives the same result for all $1 \leq \alpha < \infty$.

%%%%%%%%%%%%%%%%%%%%%%%%%%%%%%%%%
%%%%%%%%%%%%%%%%%%%%%%%%%%%%%%%%%
%%%%%%%%%%%%%%%%%%%%%%%%%%%%%%%%%
%%%%%%%%%%%%%%%%%%%%%%%%%%%%%%%%%

\subsection{Comparing $\alpha$-QSL with others quantifiers}
\label{sec:000000002D}

Let's compare the geometric QSL induced from an $\alpha$-norm with the QSL developed by Deffner and Lutz (DL)~\cite{DLQSL}. Ta\-king a pure initial state ${\rho_0} = |{\psi_0}\rangle\langle{\psi_0}|$ undergoing a certain nonunitary dynamics, DL showed the lower bound $\tau \geq \tau_{QSL}^{DL}$, with
\begin{equation}
\label{eq:0000000000015}
{\tau_{QSL}^{DL}} = \frac{1 - \langle{\psi_0}|{\rho_{\tau}}|{\psi_0}\rangle}{ \frac{1}{\tau} {\int_0^{\tau}} dt \left\| \frac{d{\rho_t}}{dt} \right\|_{\infty}} ~.
\end{equation}
Since $\text{dim}(\mathcal{H}) = 2$, we can choose an orthonormal basis $\{|{\psi_0}\rangle, |{\psi_0^{\perp}}\rangle\}$, where $\langle{\psi_0}|{\psi_0^{\perp}}\rangle = 0$. In this basis, we can define $\sigma_z = |{\psi_0}\rangle\langle{\psi_0}| - |{\psi_0^{\perp}}\rangle\langle{\psi_0^{\perp}}|$. Therefore, taking the Bloch representation for the final state, we can write $1 - \langle{\psi_0}|{\rho_{\tau}}|{\psi_0}\rangle = \frac{1}{2}({1 - {n_{z,\tau}}}) \leq \frac{1}{2} \, {d_{\text{euc}}}({\vec{n}_0},{\vec{n}_{\tau}})$, where we have used the fact that ${n_{z,t}} \leq 1$ for every $t \geq 0$, and ${d_{\text{euc}}}({\vec{n}_0},{\vec{n}_{\tau}}) = \sqrt{{n_{x,\tau}^2} + {n_{y,\tau}^2} + {(1 - {n_{z,\tau}})^2}}$ is the Euclidean distance between vectors ${\vec{n}_t} = ({n_{x,t}},{n_{y,t}},{n_{z,t}})$, and ${\vec{n}_0} = (0,0,1)$ in the chosen basis $\{|{\psi_0}\rangle,|{\psi_0^{\perp}}\rangle\}$. Therefore, using Theorem~\ref{theo:THEOREM000001}, we conclude that
\begin{equation}
\label{eq:0000000000016}
\tau_{QSL}^\alpha \geq \tau_{QSL}^{DL} ~,
\end{equation}
where we have used that ${\| d{\rho_t}/dt\|_{\infty}} = (1/2)\|d{\vec{n}_t}/dt\|$.

Importantly, Eq.~\eqref{eq:0000000000016} holds whenever we have an initial single-qubit pure state. We note that an initial mixed state can be purified, finding again the previous case. Then, the geometric QSL induced from a Schatten $\alpha$-norm [see Eq.~\eqref{eq:0000000000010}], for every $\alpha \in [1,\infty)$, is tighter than the QSL obtained by Deffner and Lutz [see Eq.~\eqref{eq:0000000000015}].

Next, we compare the $\alpha$-QSL with the speed limit developed by Campaioli {\it et al.} (CPM)~\cite{Kavan}, which in turn depends on the Hilbert-Schmidt norm. In detail, they showed that $\tau \geq {\tau_{QSL}^{CPM}}$, with
\begin{equation}
\label{eq:0000000000017}
{\tau_{QSL}^{CPM}} := \frac{ \| {\rho_{\tau}} - {\rho_0} \|_2}{{\int_{0}^{\tau}} dt \left\| \frac{d{\rho_t}}{dt} \right\|_2 } \, \tau ~.
\end{equation}
We see that Eq.~\eqref{eq:0000000000017} is a particular case of the general geome\-tric QSL induced from a Schatten $\alpha$-norm [see Eq.~\eqref{eq:0000000000009}], for $\alpha = 2$. In addition, by taking a qubit dynamics, one finds that the quantifier ${\tau_{QSL}^{CPM}}$ recovers the QSL time $\tau_{QSL}^\alpha(\tau)$ in Eq.~\eqref{eq:0000000000012}, for every value of $\alpha \in [1,\infty)$ [see Theorem~\ref{theo:THEOREM000001}]. That is, for a qubit dynamics and for every $\alpha \in [1,\infty)$, we have
\begin{equation}
\label{eq:0000000000018}
\tau_{QSL}^{CPM} =  \tau_{QSL}^{(\alpha = 2)}(\tau)=  \tau_{QSL}^\alpha(\tau) = \frac{d_{\text{euc}}({\vec{n}_0},{\vec{n}_{\tau}})}{\int_0^\tau dt \left\| \frac{d {\vec{n}_t}}{dt} \right\|} \, \tau ~.
\end{equation}
On the one hand, the CPM bound in Eq.~\eqref{eq:0000000000017} is a particular case of a more general QSL, as seen in Eq.~\eqref{eq:0000000000018}. On the other hand, Theorem~\ref{theo:THEOREM000001} implies that for a single-qubit dynamics the value assumed by the CPM bound is the same for the general case.

%%%%%%%%%%%%%%%%%%%%%%%%%%%%%%%%%
%%%%%%%%%%%%%%%%%%%%%%%%%%%%%%%%%
%%%%%%%%%%%%%%%%%%%%%%%%%%%%%%%%%
%%%%%%%%%%%%%%%%%%%%%%%%%%%%%%%%%

\subsection{Tightness condition of $\alpha$-QSL}
\label{sec:000000002E}

The geometric QSL induced from Schatten $\alpha$-norm, Eq.~\eqref{eq:0000000000012}, is tight, that is ${\tau_{QSL}^{\alpha}}/{\tau} = 1$, if and only if the dynamical path is the geodesic path, which represents the optimal path in $\mathcal{D}(\mathcal{H})$. Since the state space has an Euclidean metric induced from the Schatten norm, it follows that the geodesic is the straight line connecting the initial state $\rho_0$ to the final state $\rho_{\tau}$. We state the following result:
%%%%%%%%%%%%%%%%%%%%%%%%%%%%%%%%%%%%%%%%
%%%%%%%%%%%%%%%%%%%%%%%%%%%%%%%%%%%%%%%%
\begin{lemma}[Tightness condition in the Bloch space for the $\alpha$-QSL]
\label{lem:LEMMA000001}
Let ${\rho_t} = (1/2)(\mathbb{I} + {\vec{n}_t}\cdot\vec{\sigma}) = {\rho_{11,t}}|0\rangle\langle{0}| + (1 - {\rho_{11,t}})|1\rangle\langle{1}| + {\rho_{12,t}}|0\rangle\langle{1}| + {\rho^*_{12,t}}|1\rangle\langle{0}|$ be a given qubit state, with $\rho_{jl,t}$ being time-dependent matrix elements of the state $\rho_t$ respective to the computational basis $\{|0\rangle, |1\rangle \}$. For the general dynamics of this state, the geometric QSL induced from the Schatten $\alpha$-norm [see Eq.~\eqref{eq:0000000000007}] is tight, i.e., ${\tau_{QSL}^{\alpha}}/{\tau} = 1$, if one gets
\begin{equation}
\label{eq:0000000000019}
\frac{d{\vec{n}_t}}{dt} = \left\| \frac{d{\vec{n}_t}}{dt} \right\| \hat{r}_{0,t} ~,
\end{equation}
which implies that
\begin{equation}
\label{eq:0000000000020}
\left(\Re(d{\rho_{12,t}}/dt),-\Im(d{\rho_{12,t}}/dt), d{\rho_{11,t}}/dt\right) = {\Omega_t}~{\hat{r}_{0,t}} ~,
\end{equation}
where
\begin{equation}
\label{eq:0000000000021}
\hat{r}_{0,t} = \frac{{\vec{n}_t} - {\vec{n}_0}}{\|{\vec{n}_t} - {\vec{n}_0} \|}
\end{equation}
is the unit vector which define the direction of the line that connect the initial and the final Bloch vectors, with
\begin{equation}
\label{eq:0000000000022}
{\Omega_t} :=\sqrt{{\left|\frac{d\rho_{11,t}}{dt}\right|^2} + {\left|\frac{d\rho_{12,t}}{dt}\right|^2}} ~,
\end{equation}
and $\|d{\vec{n}_t}/dt\|$ is defined in Eq.~\eqref{eq:0000000000013}.
\end{lemma}
%%%%%%%%%%%%%%%%%%%%%%%%%%%%%%%%%%%%%%%%
%%%%%%%%%%%%%%%%%%%%%%%%%%%%%%%%%%%%%%%%
\begin{proof}
See Appendix~\ref{app:00000000D}.
\end{proof}

Therefore, due to the Theorem~\ref{theo:THEOREM000001}, the optimal dynamical path in the Bloch space is the straight line connecting the initial Bloch vector $\vec{n}_0$ to the final one $\vec{n}_{\tau}$. Indeed, the solution of Eq.~\eqref{eq:0000000000019} [or Eq.~\eqref{eq:0000000000020}] provides the geodesic path that saturates the $\alpha$-QSL bound for initial single-qubit pure states. This condition can be seen as the geometrical interpretation of the Lemma~\ref{lem:LEMMA000001}. From the algebraic point of view, the unit vector $\hat{r}_{0,t}$, which depends only on the initial and final Bloch vectors, tell us what is the signal of the rates of the populations and quantum coherences. For example, if $\hat{z}\cdot{\hat{r}_{0,t}} = (1/{\Omega_t})({d\rho_{11,t}}/{dt}) \leq 0$, then the optimal physical process need be such that $d\rho_{11,t}/dt \leq 0$, that is, a dissipative process. Then, the final and initial single-qubit states restrict the class of possible dynamical maps in which the QSL can be tight.

%%%%%%%%%%%%%%%%%%%%%%%%%%%%%%%%%
%%%%%%%%%%%%%%%%%%%%%%%%%%%%%%%%%
%%%%%%%%%%%%%%%%%%%%%%%%%%%%%%%%%
%%%%%%%%%%%%%%%%%%%%%%%%%%%%%%%%%

\section{QSL from Holder's inequality}
\label{sec:000000003}

Here we will present a family of quantum speed limits through the Holder's inequality. In addition, we will compare this one with others QSL, and analyze the necessary and sufficient condition for this QSL to be tightness. 

\subsection{$(\alpha,\beta)$-QSL}
\label{sec:000000003A}

Let $\mathcal{H}$ be a finite-dimensional Hilbert space and $\mathcal{D}(\mathcal{H})$ its associated space of quantum states. By hypothesis, this last set is a normed space as it is equipped with a family of Schatten $\alpha$-norms through which distances between quantum states can be evaluated. In addition, let's consider the dynamical map ${\Lambda_t}(\bullet)$ with ${\rho_t} = {\Lambda_t}({\rho_0})$ for some initial state ${\rho_0} \in \mathcal{D}(\mathcal{H})$. Therefore, $d{\rho_t}/dt$ is supposed be well defined for every $t \geq 0$ and the dynamical path is a continuous and smooth path in the state space.

We define the overlap between the initial and final states as $\xi(t) = \mathrm{Tr}({\rho_0}{\rho_t})$. For an initial pure state ${\rho_0} = |{\psi_0}\rangle\langle{\psi_0}|$, the overlap $\xi(t) = \langle{\psi_0}|{\rho_t}|{\psi_0}\rangle$ becomes the fidelity. Due to the Holder's inequality $|\mathrm{Tr}[Y^\dagger X]| \leq \|X\|_\alpha \|Y\|_\beta$, which holds for every two linear operators $X,Y \in L(\mathcal{H})$ and real numbers $\alpha,\beta \in [1,\infty)$, where $\frac{1}{\alpha} + \frac{1}{\beta} = 1$, one finds that the absolute value of rate of such overlap becomes
\begin{equation}
\label{eq:0000000000023}
\left|\frac{d \xi(t)}{dt}\right| = \left| \text{Tr}\left({\rho_0}\frac{d{\rho_t}}{dt}\right)\right| \leq {\|{\rho_0}\|_{\beta}} \left\| \frac{d{\rho_t}}{dt} \right\|_\alpha ~.
\end{equation}
Integrating this inequality in the interval $[0,\tau]$, we obtain  
\begin{equation}
\label{eq:0000000000024}
\frac{|\xi(\tau) - \xi(0)|}{\| {\rho_0} \|_\beta} \leq {\int_0^{\tau}} dt {\left\| \frac{d{\rho_t}}{dt} \right\|_{\alpha}} ~,
\end{equation}
where we used $|\xi(\tau) - \xi(0)| = \left| \int_0^\tau dt \, ({d \xi (t)}/{dt}) \right| \leq {\int_0^{\tau}} dt \left| {d \xi(t)}/{dt} \right|$. However, note that $\xi(0) = \mathrm{Tr}({\rho_0^2}) \in \left[\frac{1}{d},1\right ]$ is the purity of the initial state, while $\xi(\tau) = \mathrm{Tr}({\rho_{\tau}}{\rho_0}) \in [0,1]$ is the overlap between the initial and final states. Therefore, Eq.~\eqref{eq:0000000000024} implies the lower bound on the time of evolution
\begin{equation}
\label{eq:0000000000025}
\tau \geq {\tau_{QSL}^{\alpha,\beta}} ~,
\end{equation}
where we introduce the $(\alpha,\beta)$-QSL defined as
\begin{equation}
\label{eq:0000000000026}
\tau_{QSL}^{\alpha,\beta} =  \frac{\left| \text{Tr}[({\rho_{\tau}} - {\rho_0}){\rho_0}] \right|}{ {\|{\rho_0}\|_{\beta}} \, {\int_0^{\tau}} dt \left\| \frac{d {\rho_t}}{dt} \right\|_\alpha} \, \tau ~.
\end{equation}

Finally, we specialize the general QSL obtained in Eq.~\eqref{eq:0000000000026} for an initial pure state ${\rho_0} = |{\psi_0}\rangle\langle{\psi_0}|$. In this case, we find that $\text{Tr}({\rho_0^2}) = 1$ and $\text{Tr}({\rho_0}{\rho_{\tau}}) = \langle{\psi_0}|{\rho_{\tau}}|{\psi_0}\rangle \in [0,1]$. Since $\rho_0$ is a projector, we have that ${\rho_0^2} = {\rho_0}$, and therefore ${\|{\rho_0}\|_{\beta}} = 1$ for every $\beta \in[1,\infty)$. Taking all together, Eq.~\eqref{eq:0000000000026} becomes
\begin{equation}
\label{eq:0000000000027}
\frac{\tau_{QSL}^{\alpha,\beta}}{\tau} = \frac{1 - \langle \psi_0 | {\rho_{\tau}} | \psi_0 \rangle }{{\int_0^{\tau}} dt \left\| \frac{d{\rho_t}}{dt} \right\|_\alpha} ~.
\end{equation}
Noteworthy, Eq.~\eqref{eq:0000000000027} shows that there is no more dependence on the dual norm $\|\bullet\|_\beta$ of the initial state, and the bound only depends on the norm $\|\bullet\|_\alpha$. In what follows, we will investigate the case of single-qubit states.

%%%%%%%%%%%%%%%%%%%%%%%%%%%%%%%%%
%%%%%%%%%%%%%%%%%%%%%%%%%%%%%%%%%
%%%%%%%%%%%%%%%%%%%%%%%%%%%%%%%%%
%%%%%%%%%%%%%%%%%%%%%%%%%%%%%%%%%

\subsection{$(\alpha,\beta)$-QSL for qubits}
\label{sec:000000003B}

Here we address $\tau_{QSL}^{\alpha,\beta}$ in Eq.~\eqref{eq:0000000000026} for single-qubit states ${\rho_t} = (1/2)(\mathbb{I} + {\vec{n}_t}\cdot\vec{\sigma})$ that can be mixed or pure, undergoing a general physical processes over the time interval $0 \leq t \leq \tau$. In this setting, it can be verified that Eq.~\eqref{eq:0000000000026} becomes
\begin{equation}
\label{eq:0000000000028}
 \frac{\tau_{QSL}^{\alpha,\beta}}{\tau} = \frac{2^{1 - \frac{1}{\alpha}} |({\vec{n}_{\tau}} - {\vec{n}_0})\cdot{\vec{n}_0}|}{\left[ {(1 + \|{\vec{n}_0}\|)^{\beta}} + {(1 - \|{\vec{n}_0}\|)^{\beta}}\right]^{\frac{1}{\beta}} {\int_0^{\tau}} dt \left\| \frac{d {\vec{n}_t}}{dt} \right\|} ~.
\end{equation}

In particular, let's see what happens when the initial state is pure. In this case, we have $\|{\vec{n}_0}\| = 1$, and the QSL given by Eq.~\eqref{eq:0000000000028} can be written as
\begin{equation}
\label{eq:0000000000029}
\frac{\tau_{QSL}^{\alpha,\beta}}{\tau} = 2^{-\frac{1}{\alpha}} \frac{\left( 1 - {\vec{n}_{\tau}}\cdot{\vec{n}_0}\right)}{{\int_0^{\tau}} dt \left\| \frac{d {\vec{n}_t}}{dt} \right\|} ~.
\end{equation}
We note that the QSL time in Eq.~\eqref{eq:0000000000029} is independent of the dual norm $\| \cdot \|_\beta$ of the Schatten norm $\| \cdot \|_\alpha$. In fact, we can obtain Eq.~\eqref{eq:0000000000029} directly from the Eq.~\eqref{eq:0000000000027}.

%%%%%%%%%%%%%%%%%%%%%%%%%%%%%%%%%
%%%%%%%%%%%%%%%%%%%%%%%%%%%%%%%%%
%%%%%%%%%%%%%%%%%%%%%%%%%%%%%%%%%
%%%%%%%%%%%%%%%%%%%%%%%%%%%%%%%%%

\subsection{Comparing $(\alpha,\beta)$-QSL with others quantifiers}
\label{sec:000000003C}

We address a comparison between the $(\alpha,\beta)$-QSL in Eq.~\eqref{eq:0000000000026}, and the quantum speed limit obtained by del Campo {\it et al}. (CEPH)~\cite{delCampoQSL}. To do so, we fix $\alpha = \beta = 2$, and thus define the Schatten $2$-norm as a {\it bona fide} distance measure on the space of quantum states. In this case, the QSL in Eq.~\eqref{eq:0000000000026} becomes
\begin{equation}
\label{eq:0000000000030}
\frac{\tau_{QSL}^{2,2}}{\tau} = \frac{{\|{\rho_0}\|_2} \, \left| f(\tau) - f(0) \right|}{{\int_0^{\tau}} dt \left\| \frac{d{\rho_t}}{dt} \right\|_2}  ~,
\end{equation}
where $f(t) = {\text{Tr}({\rho_0}{\rho_t})}/{\|{\rho_0}\|_2^2}$ is the relative purity, with $0 \leq f(t) \leq 1$ for all $t \in [0,\tau]$, and $f(0) = 1$. In particular, introducing the parameterization $f(t) = \cos{\theta_t}$, with $0 \leq {\theta_t} \leq \pi/2$, and taking into account the fact that $|\cos{\theta_{\tau}} - 1| \geq {4 {\theta_{\tau}^2}}/{\pi^2}$, we find that Eq.~\eqref{eq:0000000000030} implies the lower bound as
\begin{equation}
\label{eq:0000000000031}
\frac{\tau_{QSL}^{2,2}}{\tau} = \frac{{\|{\rho_0}\|_2} \left| \cos{\theta_{\tau}} - 1 \right|}{{\int_0^{\tau}} dt \left\| \frac{d{\rho_t}}{dt} \right\|_2} \geq \frac{\|{\rho_0}\|_2}{{\int_0^{\tau}} dt \left\| \frac{d{\rho_t}}{dt} \right\|_2} \frac{4 {\theta_{\tau}^2}}{\pi^2} ~.
\end{equation}
The quantity ${\tau_{QSL}^{CEPH}}/{\tau} := ({4 {\theta_{\tau}^2}}/{\pi^2}){\|{\rho_0}\|_2}/ {{\int_0^{\tau}} dt \left\| {d{\rho_t}}/{dt} \right\|_2}$ in the right-hand side of Eq.~\eqref{eq:0000000000031} defines the QSL obtained by del Campo {\it et al}. (CEPH) for the nonunitary evolution in open quantum systems~\cite{delCampoQSL}. Therefore, we conclude that
\begin{equation}
\label{eq:0000000000032}
{\tau_{QSL}^{2,2}} \geq {\tau_{QSL}^{CEPH}} ~.
\end{equation}
Equation~\eqref{eq:0000000000032} shows that, for $\alpha = \beta = 2$, the QSL obtained from the Holder's inequality, Eq.~\eqref{eq:0000000000026}, is expected to be tighter than the QSL developed by del Campo {\it et al}. Here, the choice $\alpha = \beta = 2$ is necessary because ${\tau_{QSL}^{CEPH}}$ was originally derived using the Hilbert-Schmidt norm. We emphasize that our analysis is in agreement with that discussed in Ref.~\cite{delCampoQSL}.

%%%%%%%%%%%%%%%%%%%%%%%%%%%%%%%%%
%%%%%%%%%%%%%%%%%%%%%%%%%%%%%%%%%
%%%%%%%%%%%%%%%%%%%%%%%%%%%%%%%%%
%%%%%%%%%%%%%%%%%%%%%%%%%%%%%%%%%

Next, we compare the general QSL in Eq.~\eqref{eq:0000000000026} with the QSL obtained by Deffner-Lutz~\cite{DLQSL}. For an initial pure state ${\rho_0} = |{\psi_0}\rangle\langle{\psi_0}|$, the Deffner-Lutz QSL is given in Eq.~\eqref{eq:0000000000015}, while the $\tau_{QSL}^{\alpha,\beta}$ is given by Eq.~\eqref{eq:0000000000027}. Note that we have 
\begin{equation}
\label{eq:0000000000033}
 \lim_{\alpha \rightarrow \infty} \tau_{QSL}^{\alpha,\beta} = \tau_{QSL}^{DL} ~.
\end{equation}
In particular, for a single-qubit dynamics, these two QSLs can be written as
\begin{eqnarray}
\label{eq:0000000000034}
\frac{\tau_{QSL}^{DL}}{\tau} = \frac{1 - {\vec{n}_{\tau}} \cdot {\vec{n}_0}}{\int_0^\tau dt \left\| \frac{d {\vec{n}_t}}{dt} \right\|} ~,
\end{eqnarray}
and
\begin{equation}
\label{eq:0000000000035}
 \frac{\tau_{QSL}^{\alpha,\beta}}{\tau} = 2^{-\frac{1}{\alpha}} \, \frac{\left(1 - {\vec{n}_{\tau}} \cdot {\vec{n}_0}\right)}{{\int_0^{\tau}} dt \left\| \frac{d {\vec{n}_t}}{dt} \right\|} ~.
\end{equation}
In this setting, the quantifier ${\tau_{QSL}^{\alpha,\beta}}$ is related with ${\tau_{QSL}^{DL}}$ by
\begin{equation}
\label{eq:0000000000036}
{\tau_{QSL}^{\alpha,\beta}} = 2^{-\frac{1}{\alpha}}{\tau_{QSL}^{DL}} ~.
\end{equation}

In the next section, we will analyze the necessary and sufficient condition for the $(\alpha,\beta)$-QSL to be tight, providing a geometrical interpretation in terms of a optimal path in the Bloch space, and an algebraic one in terms of the populations and coherences of the density operator.

%%%%%%%%%%%%%%%%%%%%%%%%%%%%%%%%%
%%%%%%%%%%%%%%%%%%%%%%%%%%%%%%%%%
%%%%%%%%%%%%%%%%%%%%%%%%%%%%%%%%%
%%%%%%%%%%%%%%%%%%%%%%%%%%%%%%%%%

\subsection{Tightness condition of $(\alpha,\beta)$-QSL}
\label{sec:00000000D}

Here we address the QSL in Eq.~\eqref{eq:0000000000029} and analyze the necessary and sufficient condition to obtain $\tau_{QSL}^{\alpha,\beta} = \tau$, that is, the optimal dynamics. We state the following result:
%%%%%%%%%%%%%%%%%%%%%%%%%%%%%%%%%%%%%%%%
%%%%%%%%%%%%%%%%%%%%%%%%%%%%%%%%%%%%%%%%
\begin{lemma}[Tightness condition in the Bloch space for the $(\alpha,\beta)$-QSL]
\label{lem:LEMMA000002}
Let's consider an initial pure state ${\rho_0} = |{\psi_0}\rangle\langle{\psi_0}| = (1/2)(\mathbb{I} + {\vec{n}_0}\cdot\vec{\sigma})$, while its evolved state becomes ${\rho_t} = (1/2)(\mathbb{I} + {\vec{n}_t}\cdot\vec{\sigma}) = {\rho_{11,t}}|{\psi_0}\rangle\langle{\psi_0}| + (1 - {\rho_{11,t}})|{\psi_0^{\perp}}\rangle\langle{\psi_0^{\perp}}| + {\rho_{12,t}}|{\psi_0}\rangle\langle{\psi_0^{\perp}}| + {\rho^*_{12,t}}|{\psi_0^{\perp}}\rangle\langle{\psi_0}|$. Then, in the orthonormal basis $\{ |{\psi_0} \rangle, | {\psi_0^{\perp}} \rangle\}$ the $(\alpha,\beta)$-QSL [see Eq.~\eqref{eq:0000000000029}] has the following necessary and sufficient condition to be tight:
\begin{eqnarray}
\label{eq:0000000000037}
\frac{\tau_{QSL}^{\alpha,\beta}}{\tau} = 1 &\Leftrightarrow& \left\| \frac{d{\vec{n}_t}}{dt} \right\| = - {2^{-\frac{1}{\alpha}}} \left( {\vec{n}_0} \cdot \frac{d{\vec{n}_t}}{dt}\right) \nonumber \\
& & =  2^{-\frac{1}{\alpha}}\left(-\frac{d{n_{z,t}}}{dt} \right) ~,
\end{eqnarray}
with ${d {n_{z,t}}}/{dt} \leq 0$, or equivalently
\begin{eqnarray}
\label{eq:0000000000038}
\frac{\tau_{QSL}^{\alpha,\beta}}{\tau} = 1  &\Leftrightarrow& {\Omega_t} = 2^{-\frac{1}{\alpha}} \left(-\frac{d\rho_{11,t}}{dt}\right) ~,
\end{eqnarray}
where $d{\rho_{11,t}}/{dt} = d\langle{\psi_0}|{\rho_t}|{\psi_0}\rangle/dt \leq 0$, with $\left\| {d{\vec{n}_t}}/{dt} \right\|$ defined in Eq.~\eqref{eq:0000000000013}, and ${\Omega_t}$ is given in Eq.~\eqref{eq:0000000000022}.
\end{lemma}
%%%%%%%%%%%%%%%%%%%%%%%%%%%%%%%%%%%%%%%%
%%%%%%%%%%%%%%%%%%%%%%%%%%%%%%%%%%%%%%%%
\begin{proof}
See Appendix~\ref{app:00000000E}.
\end{proof}

Let's now interpret the result of the Lemma~\ref{lem:LEMMA000002}. We note that, by solving Eqs.~\eqref{eq:0000000000037} [or Eq.~\eqref{eq:0000000000038}], one one finds the geodesic curve on the space state that saturates the $(\alpha,\beta)$-QSL bound for initial single-qubit pure states. From the algebraic point of view, we see that the constraint $\tau_{QSL}^{\alpha,\beta} = \tau$ indicates that the dynamical map is dissipative, that is, the populations $\rho_{11,t} = \langle{\psi_0}|{\rho_t}|{\psi_0} \rangle$ of the instantaneous state respective to the basis $\{|{\psi_0}\rangle,|{\psi^{\perp}_0}\rangle\}$ are monotonically decreasing functions in time. This result restricts the class of quantum channels for which the $\tau_{QSL}^{\alpha,\beta}$ is tight given an initial pure state.

Finally, we address the particular case where $\alpha = \infty$ and $\beta = 1$. As discussed above, in this case and taking a pure initial state, we have that
\begin{equation}
\label{eq:0000000000039}
\lim_{\alpha \rightarrow \infty} \tau_{QSL}^{\alpha,\beta} = \tau_{QSL}^{DL} ~.
\end{equation}
From the Lemma~\ref{lem:LEMMA000002} and since $\|{\vec{n}_0}\| = 1$, we have that
\begin{eqnarray}
\label{eq:0000000000040}
\frac{\tau_{QSL}^{\alpha = \infty,\beta = 1}}{\tau} = 1 &\Leftrightarrow& \left\| \frac{d{\vec{n}_t}}{dt} \right\| = - \left({\vec{n}_0}\cdot\frac{d{\vec{n}_t}}{dt}\right) \nonumber \\
& \Leftrightarrow & \frac{d{\vec{n}_t}}{dt} = - \left\| \frac{d{\vec{n}_t}}{dt} \right\| {\vec{n}_0} ~.
\end{eqnarray}
Therefore, for $\alpha = \infty$ and $\beta = 1$, and for a pure initial state, the QSL $\tau_{QSL}^{ \infty, 1} = \tau_{QSL}^{DL}$ is tight if, and only if, the dynamical path is radial in the Bloch sphere and it has the opposite direction of the initial unit Bloch vector ${\vec{n}_0}$. Moreover, this condition can be seen as the geometric interpretation of tightness of the Deffner-Lutz QSL. 

%%%%%%%%%%%%%%%%%%%%%%%%%%%%%%%%%
%%%%%%%%%%%%%%%%%%%%%%%%%%%%%%%%%
%%%%%%%%%%%%%%%%%%%%%%%%%%%%%%%%%
%%%%%%%%%%%%%%%%%%%%%%%%%%%%%%%%%

\section{Comparing $\alpha$-QSL and $(\alpha,\beta)$-QSL}
\label{sec:000000004}

Here we investigate the relationship between $\tau_{QSL}^{\alpha}$ and $\tau_{QSL}^{\alpha,\beta}$. In order to do that, we state the following result:
%%%%%%%%%%%%%%%%%%%%%%%%%%%%%%%%%%%%%%%%
%%%%%%%%%%%%%%%%%%%%%%%%%%%%%%%%%%%%%%%%
\begin{lemma}[Inequality between $\tau_{QSL}^\alpha$ and $\tau_{QSL}^{\alpha,\beta}$]
\label{lem:LEMMA000003}
For a qubit dynamics, if ${\vec{n}_0}\cdot{\vec{n}_{\tau}} \leq 0$ and the initial state is pure, then
\begin{equation}
\label{eq:0000000000041}
\tau_{QSL}^\alpha \leq 2^{\frac{1}{2} + \frac{1}{\alpha}}\tau_{QSL}^{\alpha,\beta} ~.
\end{equation}
\end{lemma}
%%%%%%%%%%%%%%%%%%%%%%%%%%%%%%%%%%%%%%%%
%%%%%%%%%%%%%%%%%%%%%%%%%%%%%%%%%%%%%%%%
\begin{proof}
See Appendix~\ref{app:00000000F}.
\end{proof}

The inequality in Eq.~\eqref{eq:0000000000041} has an interesting geometrical interpretation. The initial pure state is represented by a point ${\vec{n}_0}$ in the surface of the Bloch sphere, since $\|{\vec{n}_0}\| = 1$. This unit vector divides the Bloch sphere into two hemispheres through the plane orthogonal to this vector. Let us call the northern hemisphere the half of the sphere that has the point ${\vec{n}_0}$, and the southern hemisphere the other half, which has the point $-{\vec{n}_0}$. In this setting, we have ${\vec{n}_{\tau}}\cdot{\vec{n}_0}\leq 0$ whenever the final vector ${\vec{n}_{\tau}}$ belongs to the southern hemisphere, and the result of the Lemma is true in these cases.

In addition, we demonstrated that the $\alpha$-QSL is tighter than the $(\alpha,\beta)$-QSL, as stated in the following result:
%%%%%%%%%%%%%%%%%%%%%%%%%%%%%%%%%%%%%%%%
%%%%%%%%%%%%%%%%%%%%%%%%%%%%%%%%%%%%%%%%
\begin{lemma}[Universal QSL]
\label{lem:LEMMA000004}
For a two-dimensional system with a initial pure state, the $\alpha$-QSL is tighter than $(\alpha,\beta)$-QSL. In other words, for every evolution time $\tau$ between a initial pure state $\rho_0$ and a final general state $\rho_\tau$ of a qubit, we have
\begin{equation}
\label{eq:0000000000042}
 \tau \geq \tau_{QSL}^{\alpha}\geq \tau_{QSL}^{\alpha,\beta}~,
\end{equation}
where $\alpha,\beta \in [1,\infty]$ and $1/\alpha + 1/\beta = 1$.
\end{lemma}
%%%%%%%%%%%%%%%%%%%%%%%%%%%%%%%%%%%%%%%%
%%%%%%%%%%%%%%%%%%%%%%%%%%%%%%%%%%%%%%%%
\begin{proof}
See Appendix~\ref{app:00000000G}.
\end{proof}

The result shown in Lemma~\ref{lem:LEMMA000004} states that the $\alpha$-QSL can be seen as a universal quantifier for qubit dynamics. Note that $\tau_{QSL}^{\alpha}$ is tighter than DL-QSL for a qubit with a pure initial state and general dynamics. Additionally, this quantifier can be viewed as a generalization of Campaioli's QSL~\cite{Kavan} for a general Schatten norm. Furthermore, we have demonstrated that the $(\alpha,\beta)$-QSL is tighter than $\tau_{QSL}^{CEPH}$~\cite{delCampoQSL}. Therefore, as a consequence of Lemma~\ref{lem:LEMMA000004}, the $\alpha$-QSL is also tighter than $\tau_{QSL}^{CEPH}$. Hence, we can conclude that for the general dynamics of a two-dimensional system with a pure initial state, the $\alpha$-QSL is the tightest quantifier, or in other words, $\tau_{QSL}^\alpha$ can be interpreted as a universal QSL.
%%%%%%%%%%%%%%%%%%%%%%%%%%%%%%%%%
%%%%%%%%%%%%%%%%%%%%%%%%%%%%%%%%%
%%%%%%%%%%%%%%%%%%%%%%%%%%%%%%%%%
%%%%%%%%%%%%%%%%%%%%%%%%%%%%%%%%%

\section{Conclusions}
\label{sec:000000005}

In this work, we discuss two families of QSLs induced by general Schatten norms. The first was obtained through the geometric approach, while the second one from the Holder's inequality. Such QSLs were compared with others speed limits found in literature, and the tightness condition for these two quantifiers was also obtained.

On the one hand, we obtain the $\alpha$-QSL within a geometric approach, which in turn provides a generalization of the QSL obtained by Campaioli {\it et al.}~\cite{Kavan} for general Schatten norms. We find that the value of such a QSL does not depend on the specific Schatten norm when we are considering a single-qubit dynamics, highlighting a certain notion of universality. Furthermore, it was noted that when the initial state is pure, such a QSL is tighter than the Deffner-Lutz QSL~\cite{DLQSL}. On the other hand, $(\alpha,\beta)$-QSL was obtained from the Holders' inequality, and can be seen as a generalization of the QSL discussed in Ref.~\cite{delCampoQSL}. Furthermore, $(\alpha,\beta)$-QSL recovers the Deffner-Lutz QSL for $\alpha \rightarrow \infty$ (operator norm), while stands proportional to this one in the case of qubit dynamics.

We investigate the necessary and sufficient conditions for $\alpha$-QSL and $(\alpha,\beta)$-QSL to be tight under general qubit dynamics. We obtain conditions with geome\-tric interpretations in terms of the dynamic path in the Bloch space, just as we obtain the constraints on the behavior of the populations and coherences of the qubit, thus reflecting an algebraic condition. Overall, the geodesic that saturates each bound is expected to be related to the generator that drives the evolution of the quantum system. On the one hand, the optimal dynamical path that saturates the $\alpha$-QSL bound is the straight line connecting initial and final states, and it is fully characterized by the generator that governs the evolution. For instance, if $\hat{z}\cdot{\hat{r}_{0,t}} = (1/{\Omega_t})(d{\rho_{11,t}}/dt) \leq 0$, it follows that the optimal physical process that saturates the $\alpha$-QSL should be a dissipative one, since that $d{\rho_{11,t}}/dt \leq 0$. On the other hand, we also address the geodesic path that saturates the $(\alpha,\beta)$-QSL bound for initial single-qubit pure states. We note that the optimal path is related to the constraint $d{\rho_{11,t}}/dt = d\langle{\psi_0}|{\rho_t}|{\psi_0}\rangle \leq 0$, which in turn signals a dissipative dynamical map in which the populations ${\rho_{11,t}} = \langle{\psi_0}|{\rho_t}|{\psi_0}\rangle$ of the instantaneous state monotonically decreases as a function of time.

Finally, we compared the $\alpha$-QSL and $(\alpha,\beta)$-QSL for the case of single-qubit systems. We found an inequality constraint between them that holds whenever the initial state is pure. This link exhibits a geometric interpretation in terms of the initial and final vectors in Bloch space. Moreover, we showed that the geometric QSL is the tightest quantifier for a general qubit dynamics with a pure initial state, and thus this quantifier can be interpreted as a universal QSL. In addition to obtaining a family of QSLs, these quantifiers generalizes previous QSLs discussed in the literature and which were developed under completely different approaches. This again reinforces the notion of universality between the different types of QSL. 

In conclusion, it seems promising to investigate the modi\-fi\-cation of our bounds within the perspective of exact and inexact speed limits, according to Refs.~\cite{arXiv:2305.03839,arXiv:2406.08584}. In principle, one can get exact or tighter QSLs bounds by removing the ``classical'' counterpart of the generator of the dynamics that commutes with the probe state~\cite{arXiv:2305.03839,arXiv:2406.08584}. This could lead to improving our QSL time in Eqs.~\eqref{eq:0000000000010} and~\eqref{eq:0000000000011}, thus providing an exact speed limit within our geometric approach. In general, the inexact speed limit is not attainable for most CPTP dynamics mainly due to the ``classical'' counterpart of the generator, i.e., that contribution that commutes with the initial state~\cite{arXiv:2305.03839,arXiv:2406.08584}. This counterpart depends on the initial and instantaneous states of the quantum system and does not drive the distinguishability of such states. This overall analysis may contribute to improve our bound in Eq.~\eqref{eq:0000000000011} respective to CPTP maps. This investigation as described deserves a study on its own and will be reported elsewhere.

%%%%%%%%%%%%%%%%%%%%%%%%%%%%%%%%%
%%%%%%%%%%%%%%%%%%%%%%%%%%%%%%%%%
%%%%%%%%%%%%%%%%%%%%%%%%%%%%%%%%%
%%%%%%%%%%%%%%%%%%%%%%%%%%%%%%%%%

\begin{acknowledgments}
This study was financed in part by the Coordena\c{c}\~{a}o de Aperfei\c{c}oamento de Pessoal de N\'{i}vel Superior--Brasil (CAPES) (Finance Code 001). D. O. S. P acknowledges the support by the Brazilian funding agencies CNPq (Grant No. 304891/2022-3), FAPESP (Grant No. 2017/03727-0) and the Brazilian National Institute of Science and Technology of Quantum Information (INCT/IQ). D. P. P. also acknow\-ledges Funda\c{c}\~{a}o de Amparo \`{a} Pesquisa e ao Desenvolvimento Cient\'{i}fico e Tecnol\'{o}gico do Maranh\~{a}o (FAPEMA).
\end{acknowledgments}

%%%%%%%%%%%%%%%%%%%%%%%%%%%%%%%%%
%%%%%%%%%%%%%%%%%%%%%%%%%%%%%%%%%
%%%%%%%%%%%%%%%%%%%%%%%%%%%%%%%%%
%%%%%%%%%%%%%%%%%%%%%%%%%%%%%%%%%

\textit{\textbf{Note added}}.--- After completion of this work, a paper has appeared~\cite{arXiv:2401.01746} reporting coherent QSL bounds based on Schatten $\alpha$-norms. We note that their approach is different from ours, and their results only apply to closed quantum systems. In addition, $\alpha$-QSL and $(\alpha,\beta)$-QSL has been not achieved in that paper.

%%%%%%%%%%%%%%%%%%%%%%%%%%%%%%%%%%%%%%%%%%%%%%%
%%%%%%%%%%%%%%%%%%%%%%%%%%%%%%%%%%%%%%%%%%%%%%%
%%%%%%%%%%%%%%%%%%---APPENDIX---%%%%%%%%%%%%%%%%%%%%%
%%%%%%%%%%%%%%%%%%%%%%%%%%%%%%%%%%%%%%%%%%%%%%%
%%%%%%%%%%%%%%%%%%%%%%%%%%%%%%%%%%%%%%%%%%%%%%%
%\clearpage
%\pagebreak
%\widetext
%\begin{center}
%\textbf{\large Supplemental Materials}
%\end{center}
%%%%%%%%%% Merge with supplemental materials %%%%%%%%%%
%%%%%%%%%% Prefix a "S" to all equations, figures, tables and reset the counter %%%%%%%%%%
\setcounter{equation}{0}
\setcounter{table}{0}
\setcounter{section}{0}
\numberwithin{equation}{section}
\makeatletter
\renewcommand{\thesection}{\Alph{section}} 
%%%%%%%%%%%%%%%%%%%%%%%%%%%%
\renewcommand{\thesubsection}{\Alph{section}.\arabic{subsection}}
\def\@gobbleappendixname#1\csname thesubsection\endcsname{\Alph{section}.\arabic{subsection}}
%\renewcommand{\p@subsection}{\@gobbleappendixname}
%%%%%%%%%%%%%%%%%%%%%%%%%%%%%
\renewcommand{\theequation}{\Alph{section}\arabic{equation}}
\renewcommand{\thefigure}{\arabic{figure}}
\renewcommand{\bibnumfmt}[1]{[#1]}
\renewcommand{\citenumfont}[1]{#1}
%\renewcommand{\thefigure}{S\arabic{figure}}
%\renewcommand{\bibnumfmt}[1]{[S#1]}
%\renewcommand{\citenumfont}[1]{S#1}
%%%%%%%%%% Prefix a "S" to all equations, figures, tables and reset the 
%counter (for example, \cite{S_RefA}) %%%%%%%%%%

\section*{Appendix}

%%%%%%%%%%%%%%%%%%%%%%%%%%%%%%%%%%
%%%%%%%%%%%%%%%%%%%%%%%%%%%%%%%%%%
%%%%%%%%%%%%%%%%%%%%%%%%%%%%%%%%%%
%%%%%%%%%%%%---APPENDIX A---%%%%%%%%%%%%%
%%%%%%%%%%%%%%%%%%%%%%%%%%%%%%%%%%
%%%%%%%%%%%%%%%%%%%%%%%%%%%%%%%%%%
%%%%%%%%%%%%%%%%%%%%%%%%%%%%%%%%%%

\section{Preliminaries and definitions}
\label{app:00000000A}

Let us start with some definitions about Schatten $\alpha$-norms, followed by the metric space structure  for the state space, and the Bloch representation of quantum states. These mathematical tools are very important in the discussion developed in this work.

\subsection{Normed spaces and Schatten norms}
\label{app:00000000A1}

Let us consider the Hilbert space $\mathcal{H}$ and we will denote $L(\mathcal{H})$ as the set of every linear transformations on this space. In fact, $L(\mathcal{H})$ is a complex vectorial space and this one can be transformed in a normed space with the help of the Schatten $\alpha$-norms:
\begin{definition}[Schatten norms]
\label{def:DEF000000001}
\label{Schatten norms}
For every linear operator $A \in L(\mathcal{H})$ acting on a Hilbert space $\mathcal{H}$, we define it's \textit{Schatten $\alpha$-norm} by: 
\begin{equation}
\label{eq:SM0000000001}
\| A \|_{\alpha} := (\mathrm{Tr}\{(\sqrt{A^\dagger A})^{\alpha}\})^{\frac{1}{\alpha}} ~,
\end{equation}
where $\alpha \in [1,\infty)$.
\end{definition}

Some values of $\alpha$ for the Schatten norms defined in the Definition~\ref{def:DEF000000001} are very common in the study of the space of quantum states. In particular, for $\alpha = 1$ we have the trace norm (or $L_1$-norm), for $\alpha = 2$ we have the Hilbert-Schmidt (or $L_2$-norm) norm and finally for the limit $\alpha \rightarrow \infty$ we have the operator norm (or spectral norm). We can compute these Schatten norms through the singular values. We point out that the singular values of some linear operator $A \in L(\mathcal{H})$ are given by the eigenvalues of the positive operator $\sqrt{A^\dagger A}$. In this way, the Schatten $\alpha$-norm of the operator $A$ can be given by
\begin{equation}
\label{eq:SM0000000002}
\| A \|_{\alpha} = \left(\sum_{i = 1}^\text{r} \sigma(A)_{i}^\alpha\right)^\frac{1}{\alpha} ~,
\end{equation}
where $\sigma(A)_i$ are the singular values of the operator $A$, and $r = \text{rank}(A)$ is the rank of $A$. Through Eq.~\eqref{eq:SM0000000002}, we have that the trace norm, Hilbert-Schmidt norm, and operator norm can be written as
\begin{align}
\label{eq:SM0000000003}
{\| A \|_1} &= {\sum_{j = 1}^r} \, {\sigma(A)_j} ~,\\
{\| A \|_2} &= \sqrt{{\sum_{j = 1}^r} \, {\sigma(A)_j^2}} ~,\\
{\| A \|_{\infty}} &= {\max_j} \{ {\sigma(A)_j} \} ~,
\end{align}
respectively. Now, since the singular values are the eigenva\-lues of $\sqrt{A^\dagger A}$, if the operator $A$ is Hermitian (self-adjoint), $A^\dagger = A$, then the eigenvalues of $A$, denoted by $\lambda(A)$, are related with the singular values $\{{\sigma(A)_j}\}_{j = 1,\ldots,r}$ as follows
\begin{equation}
\label{eq:SM0000000004}
{\sigma(A)_j} = |{\lambda(A)_j}| ~.
\end{equation}
In particular, note that if the operator $A$ is positive semi-definite, then the singular values are equals to the eigenvalues. 

The Schatten $\alpha$-norms have a large number of interesting properties. In particular, we have the follow result: 
\begin{theorem}[Duality of the Schatten norms]
\label{theosuppmat:THEOREMSUPPMATERIAL00000001}
For every linear operator $A$, the Schatten norm can be given as a \textit{supremum}:
\begin{equation}
\label{eq:SM0000000005}
\|A\|_\alpha = \sup_{\chi \neq 0}\left\{ |\mathrm{Tr}[\chi^\dagger A]|: \|\chi\|_\beta \leq 1, \frac{1}{\alpha}+\frac{1}{\beta} = 1  \right\} ~.
 \end{equation}
\end{theorem}

In particular, due to the Theorem~\ref{theosuppmat:THEOREMSUPPMATERIAL00000001}, we have the Holder's inequality, defined as
\begin{equation}
\label{eq:SM0000000006}
|\mathrm{Tr}[Y^\dagger X]| \leq \|X\|_\alpha \|Y\|_\beta ~,
\end{equation}
which holds for every two linear operators $X,Y \in L(\mathcal{H})$ and for every real numbers $\alpha,\beta \in [1,\infty)$ where $\frac{1}{\alpha} + \frac{1}{\beta} = 1$. In this sense, the norms $\| \cdot\|_\alpha$ and $\| \cdot\|_\beta$ are said to be \textit{dual} to each other. The Holder's inequality in Eq.~\eqref{eq:SM0000000006} is the fundamental inequality that implied the $(\alpha,\beta)$-QSL discussed in the main text.

%%%%%%%%%%%%%%%%%%%%%%%%%%%%%%%%%
%%%%%%%%%%%%%%%%%%%%%%%%%%%%%%%%%
%%%%%%%%%%%%%%%%%%%%%%%%%%%%%%%%%
%%%%%%%%%%%%%%%%%%%%%%%%%%%%%%%%%

\subsection{Metric spaces and space of quantum states}
\label{app:00000000A2}

Since $L(\mathcal{H})$ is in fact a complex vectorial space, we can write its elements in a particular chosen basis. This idea is the core for the Bloch representation of quantum states (density operators). In this way, we can see density operators as vectors in the Euclidean space, namely the Bloch space. Here, let us denote $\mathcal{D}(\mathcal{H})$ as the set of every quantum state $\rho$ associated to the Hilbert space $\mathcal{H}$. Obviously $ \mathcal{D}(\mathcal{H}) \subseteq L(\mathcal{H}) $, since every density operator is also a linear operator acting on the Hilbert space $\mathcal{H}$. Then, the set $\mathcal{D}(\mathcal{H})$ will be called state space of the Hilbert space $\mathcal{H}$. In fact, the space of states equipped with a Schatten $\alpha$-norm is not only a normed space, but it is also a metric space. Remember that a metric space is an ordered pair $(M,D)$ where $M$ is a set and $D$ is a metric on $M$, where $D:M \times M \rightarrow \mathbb{R}$ is a function such that, for every $x, y, z \in M$, we have: 
\begin{itemize}
\item \textbf{Positivity semi-definite:} $D(x,y) \geq 0$, with $D(x,y) = 0$ if and only if $x = y$.
\item \textbf{Symmetry:} $D(x,y) = D(y,x)$.
\item \textbf{Triangle inequality:} $D(x,z) \leq D(x,y)+D(y,z)$.
\end{itemize}

In this way, taking the set $M = \mathcal{D}(\mathcal{H})$, and considering the $\alpha$-norm in the state space, we can define the metric
\begin{equation}
\label{eq:SM0000000007}
 D_{\alpha}(\rho,\eta) = \|\rho - \eta\|_{\alpha} ~,
\end{equation}
and then the state space becomes an Euclidean metric space. Therefore, an normed state space $\mathcal{D}(\mathcal{H})$ equipped with a Schatten $\alpha$-norm becomes, in fact, an authentic metric space. We can work with elements of the state space in a very useful fashion, called Bloch representation. This approach provides a useful geometric interpretation for two-dimensional quantum systems.

%%%%%%%%%%%%%%%%%%%%%%%%%%%%%%%%%
%%%%%%%%%%%%%%%%%%%%%%%%%%%%%%%%%
%%%%%%%%%%%%%%%%%%%%%%%%%%%%%%%%%
%%%%%%%%%%%%%%%%%%%%%%%%%%%%%%%%%

\subsection{General qudit systems and Bloch representation}
\label{app:00000000A3}

Now, let us consider a finite dimensional Hilbert space $\mathcal{H}$ with dimension $d$, and let $\{F_\alpha\}_{\alpha = 0}^{d^2-1}$ be an orthonormal traceless Hermitian basis in $L(\mathcal{H})$, that is $\mathrm{Tr}(F_\alpha F_\beta) = \delta_{\alpha,\beta}$, together with $F_0 = \mathds{1}/\sqrt{d}$ and $F_\alpha^\dagger = F_\alpha$, as well as $\mathrm{Tr}(F_\alpha) = 0$ for every $\alpha = 1,2,...,d^2-1$. In this way, any quantum state (density operator) $\rho$ of the Hilbert Space $\mathcal{H}$ can be represented by 
\begin{equation}
\label{eq:SM0000000008}
{\rho_t} = \frac{1}{d} \left( \mathds{1} + {\vec{x}_t}\cdot\vec{F} \right) ~,
\end{equation}
where ${\vec{x}_t}\cdot{\vec{F}} = {\sum_{i = 1}^{d^2-1}}\,{x_{i,t}} {F_i}$ sets for the standard inner product, while ${x_{j,t}} = \text{Tr}({F_j}{\rho_t})$, where ${\vec{x}_t} = ({x_{1,t}},...,{x_{{d^2} - 1,t}}) \in {\mathbb{R}^{d^2-1}}$, and $\vec{F} = (F_1,...,F_{d^2-1})$. In this setting, every density operator $\rho$ can be represented by a vector $\vec{x}$ belonging to a subset of the Euclidean space $\mathbb{R}^{d^2-1}$. The subset of $\mathbb{R}^{d^2-1}$ of every accessible states is called Bloch space. For a qubit system, the set of accessible states has an spherical geometry, and it is called Bloch sphere. For systems in which the dimension $d$ is greater than two, the geometry of the Bloch space can becomes very complicated. Actually, note that the operators $F_{\alpha}$ are the generators of the $SU(d)$ Lie algebra. 

%%%%%%%%%%%%%%%%%%%%%%%%%%%%%%%%%
%%%%%%%%%%%%%%%%%%%%%%%%%%%%%%%%%
%%%%%%%%%%%%%%%%%%%%%%%%%%%%%%%%%
%%%%%%%%%%%%%%%%%%%%%%%%%%%%%%%%%

\section{Single-qubit systems}
\label{app:00000000B}

%%%%%%%%%%%%%%%%%%%%%%%%%%%%%%%%%
%%%%%%%%%%%%%%%%%%%%%%%%%%%%%%%%%
%%%%%%%%%%%%%%%%%%%%%%%%%%%%%%%%%
%%%%%%%%%%%%%%%%%%%%%%%%%%%%%%%%%

\subsection{General single-qubit state and the Bloch representation}
\label{app:00000000B1}

For a two-dimensional Hilbert space, we have the Pauli matrix operators as a basis for $\mathcal{D}(\mathcal{H})$. In this case, we have ${\rho_t} = (1/2)\left( \mathbb{I} + {\vec{n}_t}\cdot\vec{\sigma}\right)$, where ${\vec{n}_t} = ({n_{x,t}},{n_{y,t}},{n_{z,t}}) \in \mathbb{R}^3$ is a time-dependent three-dimensional vector, and $\vec{\sigma} = (\sigma_x,\sigma_y,\sigma_z)$ is the vector of Pauli matrices. We remember that the density matrix is Hermitian, ${\rho_t^{\dagger}} = {\rho_t}$, positive semi-definite, ${\rho_t} \geq 0$, with trace equal to the unit, $\text{Tr}({\rho_t}) = 1$, for all $t \geq 0$. Here, we set the eigenbasis $\{ |0\rangle, |{1}\rangle  \}$ of the operator $\sigma_z$ for the two-dimensional Hilbert space $\mathcal{H}$, with ${\sigma_z}|\ell\rangle =  {(-1)^{\ell}}|\ell\rangle$, and $\ell = \{0,1\}$. With respect to this eigenbasis, it follows that the aforementioned single-qubit state is written as follows
\begin{equation}
\label{eq:SM0000000009}
{\rho_t} = \frac{1}{2}\begin{pmatrix} 1 + {n_{z,t}} & {n_{x,t}} - i {n_{y,t}} \\ {n_{x,t}} + i {n_{y,t}} & 1 - {n_{z,t}} \end{pmatrix} ~.
\end{equation}
The eigenvalues of the density matrix in Eq.~\eqref{eq:SM0000000009} can be written as
\begin{equation}
\label{eq:SM0000000010}
 \lambda_{\pm}(t) = \frac{1}{2} (1 \pm \| {\vec{n}_t} \|) ~,
\end{equation}
where $\|{\vec{n}_t}\| = \sqrt{{n_{x,t}^2} + {n_{y,t}^2} + {n_{z,t}^2}}$ is the Euclidean norm of the vector ${\vec{n}_t}$, with $\| {\vec{n}_t} \| \leq 1$. The Bloch space is a subset of the three-dimensional sphere in the Euclidean space $\mathbb{R}^3$. We note that the density matrix $\rho_t$ exhibits the singular va\-lues $\sigma_{\pm} = (1/2) (1 \pm \| {\vec{n}_t} \|)$ [see Eq.~\eqref{eq:SM0000000004}]. Finally, since we know the sin\-gu\-lar values of $\rho_t$, we can compute its Schatten $\alpha$-norm as [see Eq.~\eqref{eq:SM0000000002}]
\begin{equation}
\label{eq:SM0000000011}
\| {\rho_t} \|_\alpha = \frac{1}{2} \left[(1 - \| {\vec{n}_t} \|)^\alpha+(1+\| {\vec{n}_t} \|)^\alpha\right]^{\frac{1}{\alpha}} ~.
\end{equation}

Note that both populations and coherences of the density matrix $\rho_t$ can be written in terms of the elements $\{ {n_{j,t}} \}_{j = x, y, z}$ of the Bloch vector ${\vec{n}_t}$, and vice-versa. On the one hand, by writing the density matrix in the eigenbasis $\{ |0\rangle, |1\rangle \}$, we have that ${\rho_t} = {p_t}|0\rangle\langle{0}| + (1 - {p_t})|{1}\rangle\langle{1}| + {q_t}|0\rangle\langle{1}| + {q^*_t}|1\rangle\langle{0}|$, where ${p_t} = \langle{0}|{\rho_t}|{0}\rangle \in [0,1]$ is the population of the state $|0\rangle$, and ${q_t} = \langle{0}|{\rho_t}|{1}\rangle$ is the coherence between vector states $|0\rangle$ and $|1\rangle$. On the other hand, from Eq.~\eqref{eq:SM0000000009}, we find that ${p_t} = (1/2)(1 + {n_{z,t}})$, and ${q_t} = (1/2)({n_{x,t}} - i \, {n_{y,t}})$, which readily implies the relations ${n_{z,t}} = 2{p_t} - 1$, ${n_{x,t}} = 2 \Re({q_t})$, and ${n_{y,t}} = -2 \Im({q_t})$.

Next, we suppose that the qubit state ${\rho_0}$ undergoes a dynamical evolution governed by the time-differentiable dynamical map ${\Lambda_t}$. In the Bloch representation, taking the eigenbasis of $\sigma_z$ for the two-dimensional Hilbert space $\mathcal{H}$, we find that
\begin{align}
\label{eq:SM0000000012}
 \frac{d{\rho_t}}{dt} &= \frac{1}{2} \, \frac{d{\vec{n}_t}}{dt} \cdot \vec{\sigma} \nonumber\\
 &= \frac{1}{2}\begin{pmatrix} d{n_{z,t}}/dt & d{n_{x,t}}/dt - i\, d{n_{y,t}}/dt \\ d{n_{x,t}}/dt + i\, d{n_{y,t}}/dt & - d{n_{z,t}}/dt  \end{pmatrix} ~.
\end{align}
The eigenvalues of the operator ${d\rho_t}/dt$ are given by $\gamma_{\pm} = \pm (1/2)\left\| {d {\vec{n}_t}}/{dt} \right\|$, where we define the Euclidean norm ${\left\| {d{\vec{n}_t}}/{dt} \right\|^2} = {\sum_{j = x,y,z}}\, {({d{n_{j,t}}}/dt)^2}$. The singular values of this operator are given by $\omega_1 = \omega_2 = (1/2)\left\| {d {\vec{n}_t}}/{dt} \right\|$. Hence, Eq.~\eqref{eq:SM0000000002} implies the following result
\begin{equation}
\label{eq:SM0000000013}
\left\| \frac{d {\rho_t}}{dt}  \right\|_\alpha = 2^{-1 + \frac{1}{\alpha}} \left\| \frac{d {\vec{n}_t}}{dt}  \right\| ~.
\end{equation}
Equation~\eqref{eq:SM0000000013} means that the Schatten $\alpha$-norm $\left\|{d{\rho_t}}/{dt}\right\|_{\alpha}$ in the state space is proportional to the Euclidean norm $\left\|{d{\vec{n}_t}}/{dt}\right\|$ in the Bloch space. 

%%%%%%%%%%%%%%%%%%%%%%%%%%%%%%%%%
%%%%%%%%%%%%%%%%%%%%%%%%%%%%%%%%%
%%%%%%%%%%%%%%%%%%%%%%%%%%%%%%%%%
%%%%%%%%%%%%%%%%%%%%%%%%%%%%%%%%%

\subsection{The space of single-qubit states as a metric space}
\label{app:00000000B2}

Let us compute the distance ${D_{\alpha}}({\rho_t},{\eta_t}) = \|{\rho_t} - {\eta_t}\|_{\alpha}$ [see Eq.~\eqref{eq:SM0000000007}] between the single-qubit states ${\rho_t}, {\eta_t} \in \mathcal{D}(\mathcal{H})$, with ${\rho_t} = (1/2)(\mathbb{I} + {\vec{n}_t}\cdot\vec{\sigma})$, and ${\eta_t} = (1/2)(\mathbb{I} + {\vec{m}_t}\cdot\vec{\sigma})$. First, we note that
\begin{equation}
\label{eq:SM0000000014}
{\rho_t} - {\eta_t} = \frac{1}{2} ({\vec{n}_t} - {\vec{m}_t})\cdot\vec{\sigma} ~,
\end{equation}
where ${\vec{n}_t} = ({n_{x,t}},{n_{y,t}},{n_{z,t}})$, and ${\vec{m}_t} = ({m_{x,t}}, {m_{y,t}}, {m_{z,t}})$ are the Bloch vectors of $\rho_t$ and $\eta_t$, respectively. We see that Eqs.~\eqref{eq:SM0000000012} and~\eqref{eq:SM0000000014} are similar each other for the single-qubit dynamics. Hence, the eigenvalues of the operator ${\rho_t} - {\eta_t}$ and the respective singular values can be obtained in the same way. In particular, the Schatten $\alpha$-norm of the operator ${\rho_t} - {\eta_t}$ [that is, the distance $D_\alpha({\rho_t}, {\eta_t})$] is given by
\begin{equation}
\label{eq:SM0000000015}
{D_\alpha}({\rho_t}, {\eta_t}) = {\| {\rho_t} - {\eta_t} \|_{\alpha}} = 2^{-1 + \frac{1}{\alpha}} {d_{\text{euc}}}({\vec{n}_t}, {\vec{m}_t}) ~,
\end{equation}
where
\begin{align}
\label{eq:SM0000000016}
&{d_{\text{euc}}}({\vec{n}_t},{\vec{m}_t}) = \nonumber\\
&\sqrt{{({n_{x,t}} - {m_{x,t}})^2} + {({n_{y,t}} - {m_{y,t}})^2} + {({n_{z,t}} - {m_{z,t}})^2}} ~.
\end{align}
Equation~\eqref{eq:SM0000000015} means that the distance $D_{\alpha}$ in the state space $\mathcal{D}(\mathcal{H})$ induced from the $\alpha$-norm is proportional to the Euclidean distance $d_{\text{euc}}$ in the Bloch space.

%%%%%%%%%%%%%%%%%%%%%%%%%%%%%%%%%
%%%%%%%%%%%%%%%%%%%%%%%%%%%%%%%%%
%%%%%%%%%%%%%%%%%%%%%%%%%%%%%%%%%
%%%%%%%%%%%%%%%%%%%%%%%%%%%%%%%%%

\section{Proof of Theorem~\ref{theo:THEOREM000001}}
\label{app:00000000C}

%%%%%%%%%%%%%%%%%%%%%%%%%%%%%%%%%
%%%%%%%%%%%%%%%%%%%%%%%%%%%%%%%%%
%%%%%%%%%%%%%%%%%%%%%%%%%%%%%%%%%
%%%%%%%%%%%%%%%%%%%%%%%%%%%%%%%%%

The quantum speed limit time induced by the $\alpha$-Schatten norm, for a time of evolution $\tau$, is given by
\begin{equation}
\label{eq:SM0000000017}
{\tau_{QSL}^{\alpha}} (\tau) := \frac{ \| {\rho_{\tau}} - {\rho_0} \|_\alpha}{ {\int_{0}^{\tau}} dt \left\| \frac{d{\rho_t}}{dt} \right\|_\alpha} \, \tau ~.
\end{equation}
Given the single-qubit states ${\rho_0} = (1/2)(\mathbb{I} + {\vec{n}_0}\cdot\vec{\sigma})$ and ${\rho_{\tau}} = (1/2)(\mathbb{I} + {\vec{n}_{\tau}}\cdot\vec{\sigma})$, one can verify the identity [see Appendix~\ref{app:00000000B}]
\begin{equation}
\label{eq:SM0000000018}
\| {\rho_{\tau}} - {\rho_0} \|_\alpha = 2^{-1 + \frac{1}{\alpha}} \, {d_{\text{euc}}}({\vec{n}_0},{\vec{n}_{\tau}}) ~,
\end{equation}
with the Euclidean distance defined as
\begin{equation}
\label{eq:SM0000000019}
d_{\text{euc}}({\vec{n}_0}, {\vec{n}_{\tau}}) = \sqrt{ {\sum_{j = x,y,z}} ({n_{j,\tau}} - {n_{j,0}})^2} ~.
\end{equation}
We also note that
\begin{equation}
\label{eq:SM0000000020}
\left\| \frac{d {\rho_t}}{dt}  \right\|_\alpha = 2^{-1 + \frac{1}{\alpha}} \left\| \frac{d {\vec{n}_t}}{dt}  \right\| ~,
\end{equation}
with
\begin{equation}
\label{eq:SM0000000021}
\left\| \frac{d{\vec{n}_t}}{dt} \right\| = \sqrt{{\sum_{j = x,y,z}}\,{\left(\frac{d{n_{j,t}}}{dt}\right)^2}} ~.
\end{equation}
Hence, by combining Eqs.~\eqref{eq:SM0000000017},~\eqref{eq:SM0000000018}, and~\eqref{eq:SM0000000020}, one gets the result
\begin{equation}
\label{eq:SM0000000022}
\tau_{QSL}^\alpha(\tau) = \frac{d_{\text{euc}}({\vec{n}_0},{\vec{n}_{\tau}})}{{\int_0^{\tau}} dt \left\| \frac{d{\vec{n}_t}}{dt} \right\|} \, \tau ~,
\end{equation}
which concludes the proof of Theorem~\ref{theo:THEOREM000001}.

%%%%%%%%%%%%%%%%%%%%%%%%%%%%%%%%%
%%%%%%%%%%%%%%%%%%%%%%%%%%%%%%%%%
%%%%%%%%%%%%%%%%%%%%%%%%%%%%%%%%%
%%%%%%%%%%%%%%%%%%%%%%%%%%%%%%%%%

\section{Proof of Lemma~\ref{lem:LEMMA000001}}
\label{app:00000000D}

%%%%%%%%%%%%%%%%%%%%%%%%%%%%%%%%%
%%%%%%%%%%%%%%%%%%%%%%%%%%%%%%%%%
%%%%%%%%%%%%%%%%%%%%%%%%%%%%%%%%%
%%%%%%%%%%%%%%%%%%%%%%%%%%%%%%%%%

Here we discuss the conditions in which the $\alpha$-QSL in Eq.~\eqref{eq:SM0000000022} for single-qubit states is tight, i.e., ${\tau_{QSL}^\alpha}(\tau)/{\tau}$ = 1. To do so, one should satisfies the constraint as follows
\begin{equation}
\label{eq:SM0000000023}
{d_{\text{euc}}}({\vec{n}_0},{\vec{n}_{\tau}}) = {\int_0^{\tau}} dt \left\| \frac{d{\vec{n}_t}}{dt} \right\| ~,
\end{equation}
where the Euclidean distance ${d_{\text{euc}}}({\vec{n}_0},{\vec{n}_{\tau}})$ is defined in Eq.~\eqref{eq:SM0000000019}. By differentiating both sides of Eq.~\eqref{eq:SM0000000023} with respect to the parameter $\tau$, one obtains the result
\begin{equation}
\label{eq:SM0000000024}
\left({\hat{r}_{0,t}} \cdot \frac{d{\vec{n}_t}}{dt} \right)_{t = \tau} = \left(\left\| \frac{d {\vec{n}_t}}{dt} \right\|\right)_{t = \tau} ~,
\end{equation}
where we have defined the unit vector
\begin{equation}
\label{eq:SM0000000025}
\hat{r}_{0,t} := \frac{{\vec{n}_{t}} - {\vec{n}_0}}{\|{\vec{n}_{t}} - {\vec{n}_0} \|} ~.
\end{equation}
We note that Eq.~\eqref{eq:SM0000000024} holds for all $t \geq 0$ (or even $\tau \geq 0$), and from now on we deliberately work with such an equation written in terms of the parameter $t$. In addition, Eq.~\eqref{eq:SM0000000024} can be solved in terms of the quantity $\left\| {d {\vec{n}_t}}/{dt} \right\|$ as long as one chooses the vector ${d {\vec{n}_t}}/{dt}$ as follows
\begin{equation}
\label{eq:SM0000000026}
 \frac{d {\vec{n}_t}}{dt} = \left\| \frac{d {\vec{n}_t}}{dt} \right\| {\hat{r}_{0,t}} ~,
\end{equation}
where we used the fact that $\hat{r}_{0,t}\cdot\hat{r}_{0,t} = 1$, for all $t \geq 0$. Therefore, Eq.~\eqref{eq:SM0000000026} means that ${\tau_{QSL}^{\alpha}} = \tau$ if and only if the dynamical path in the Bloch space is given by the straight line connecting the initial Bloch vector $\vec{n}_0$ to the final one ${\vec{n}_{\tau}}$. Next, we note that Eq.~\eqref{eq:SM0000000026} can be recasted in terms of the populations and quantum coherences of the evolved state $\rho_t$ respective to the computational eigenbasis $\{ |0\rangle, |1\rangle \}$. To do so, we use the fact that ${n_{z,t}} = 2{\rho_{11,t}} - 1$, ${n_{x,t}} = 2\Re(\rho_{12,t})$, and ${n_{y,t}} = -2\Im({\rho_{12,t}})$ for single-qubit states. In this case, one readily obtains
\begin{align}
\label{eq:SM0000000027}
\frac{d {\vec{n}_t}}{dt} &= \left(d{n_{x,t}}/dt, d{n_{y,t}}/dt, d{n_{z,t}}/dt\right) \nonumber\\
&= 2\left(\Re({d\rho_{12,t}}/{dt}),-\Im({d\rho_{12,t}}/{dt}), {d\rho_{11,t}}/{dt}\right) ~.
\end{align}
Next, Eq.~\eqref{eq:SM0000000027} implies that the Euclidean norm $\|{d {\vec{n}_t}}/{dt}\|$ is written as follows
\begin{equation}
\label{eq:SM0000000028}
\left\| \frac{d{\vec{n}_t}}{dt} \right\| = \sqrt{{\sum_{j = x,y,z}}\,{\left(\frac{d{n_{j,t}}}{dt}\right)^2}} = 2{\Omega_t} ~,
\end{equation}
where we have defined the quantity 
\begin{equation}
\label{eq:SM0000000029}
{\Omega_t} := \sqrt{ \, \left| \frac{d\rho_{11,t}}{dt} \right|^2+\left| \frac{d\rho_{12,t}}{dt} \right|^2}
\end{equation}
Therefore, by combining Eqs.~\eqref{eq:SM0000000026},~\eqref{eq:SM0000000027} and~\eqref{eq:SM0000000028}, we find that
\begin{equation}
\label{eq:SM0000000030}
\left(\Re({d\rho_{12,t}}/{dt}),-\Im({d\rho_{12,t}}/{dt}), {d\rho_{11,t}}/{dt}\right) =  {\Omega_t}\, \hat{r}_{0,t} ~,
\end{equation}
and one concludes the proof of Lemma~\ref{lem:LEMMA000001}.

%%%%%%%%%%%%%%%%%%%%%%%%%%%%%%%%%
%%%%%%%%%%%%%%%%%%%%%%%%%%%%%%%%%
%%%%%%%%%%%%%%%%%%%%%%%%%%%%%%%%%
%%%%%%%%%%%%%%%%%%%%%%%%%%%%%%%%%

\section{Proof of Lemma~\ref{lem:LEMMA000002}}
\label{app:00000000E}

%%%%%%%%%%%%%%%%%%%%%%%%%%%%%%%%%
%%%%%%%%%%%%%%%%%%%%%%%%%%%%%%%%%
%%%%%%%%%%%%%%%%%%%%%%%%%%%%%%%%%
%%%%%%%%%%%%%%%%%%%%%%%%%%%%%%%%%

In order to achieve the tighter $(\alpha,\beta)$-QSL, i.e., ${\tau_{QSL}^{\alpha,\beta}}/{\tau} = 1$, we require that
\begin{equation}
\label{eq:SM0000000031}
{2^{-\frac{1}{\alpha}}} (1 - {\vec{n}_{\tau}}\cdot{\vec{n}_0}) = {\int_0^{\tau}} dt \left\| \frac{d{\vec{n}_t}}{dt} \right\| ~.
\end{equation}
We note that Eq.~\eqref{eq:SM0000000031} is satisfied as long as one sets the Euclidean norm as
\begin{equation}
\label{eq:SM0000000032}
\left\| \frac{d{\vec{n}_t}}{dt} \right\| = - {2^{-\frac{1}{\alpha}}} \left({\vec{n}_0} \cdot \frac{d{\vec{n}_t}}{dt}\right) ~.
\end{equation}
In particular, we hereafter fix the orthonormal basis $\{ |{\psi_0}\rangle, |{\psi_0^{\perp}} \rangle\}$. In this setting, one gets the vector ${\vec{n}_0} = (0,0,1)$, which in turn implies that ${\vec{n}_0}\cdot({d{\vec{n}_t}}/{dt}) = d{n_{z,t}}/dt $.  Hence, Eq.~\eqref{eq:SM0000000032} becomes
\begin{equation}
\label{eq:SM000000033}
\left\| \frac{d{\vec{n}_t}}{dt} \right\| = {2^{-\frac{1}{\alpha}}} \left(- \frac{d{n_{z,t}}}{dt} \right) ~.
\end{equation}
We note that, since $\left\| {d}{\vec{n}_t}/{dt} \right\| \geq 0$, we thus necessarily have that $d{n_{z,t}}/dt \leq 0$ in Eq.~\eqref{eq:SM000000033}. From Appendix~\ref{app:00000000D}, we know that $d{n_{z,t}}/dt = 2\, d{\rho_{11,t}}/dt = 2\, d\langle{\psi_0}|{\rho_t}|{\psi_0}\rangle/dt$ [see Eq.~\eqref{eq:SM0000000027}], and also $\left\| {d{\vec{n}_t}}/{dt} \right\| = 2{\Omega_t}$ [see Eqs.~\eqref{eq:SM0000000028} and~\eqref{eq:SM0000000029}]. Finally, by substituting these results in Eq.~\eqref{eq:SM000000033}, one obtains that
\begin{equation}
\label{eq:SM0000000034}
{\Omega_t} = {2^{-\frac{1}{\alpha}}} \left(- \frac{d{\rho_{11,t}}}{dt} \right) ~,
\end{equation}
with 
\begin{equation}
\label{eq:SM000000003402}
\frac{d{\rho_{11,t}}}{dt} = \frac{d}{dt} \langle{\psi_0}|{\rho_t}|{\psi_0}\rangle \leq 0 ~.
\end{equation}
In conclusion, Eq.~\eqref{eq:SM0000000034} proves Lemma~\ref{lem:LEMMA000002}.

%%%%%%%%%%%%%%%%%%%%%%%%%%%%%%%%%
%%%%%%%%%%%%%%%%%%%%%%%%%%%%%%%%%
%%%%%%%%%%%%%%%%%%%%%%%%%%%%%%%%%
%%%%%%%%%%%%%%%%%%%%%%%%%%%%%%%%%

\section{Proof of Lemma~\ref{lem:LEMMA000003}}
\label{app:00000000F}

%%%%%%%%%%%%%%%%%%%%%%%%%%%%%%%%%
%%%%%%%%%%%%%%%%%%%%%%%%%%%%%%%%%
%%%%%%%%%%%%%%%%%%%%%%%%%%%%%%%%%
%%%%%%%%%%%%%%%%%%%%%%%%%%%%%%%%%

We consider a two-level system initialized in a single-qubit pure state that undergoes a general evolution. In this setting, given the dynamics between states $\rho_0 = (1/2)(\mathbb{I} + {\vec{n}_0}\cdot\vec{\sigma})$ and $\rho_{\tau} = (1/2)(\mathbb{I} + {\vec{n}_{\tau}}\cdot\vec{\sigma})$, we proved the geometric $\alpha$-QSL time as 
\begin{equation}
\label{eq:SM0000000035}
\frac{\tau_{QSL}^{\alpha}}{\tau} = \frac{d_{\text{euc}}({\vec{n}_0},{\vec{n}_{\tau}})}{\int_0^\tau dt \left\| {d {\vec{n}_t}}/{dt} \right\|} ~,
\end{equation}
and also the $(\alpha,\beta)$-QSL time as 
\begin{equation}
\label{eq:SM0000000036}
\frac{\tau_{QSL}^{\alpha,\beta}}{\tau} = {2^{-\frac{1}{\alpha}}} \frac{\left( 1 - {\vec{n}_{\tau}}\cdot{\vec{n}_0}\right)}{{\int_0^{\tau}} dt \left\| {d {\vec{n}_t}}/{dt} \right\|}
\end{equation}
that is based on Holder's inequality, where ${d_{\text{euc}}}({\vec{n}_0},{\vec{n}_{\tau}}) = \|{\vec{n}_{\tau}} - {\vec{n}_0}  \|$ is the Euclidean distance between the initial and final states, with $\|\vec{x}\| = \sqrt{\vec{x}\cdot\vec{x}}$ for every vector $\vec{x}$ in the Bloch space. We note that the Euclidean distance can be written as follows
\begin{equation}
\label{eq:SM0000000037}
d_{\text{euc}}({\vec{n}_0},{\vec{n}_{\tau}}) =  \sqrt{{\|{\vec{n}_0}\|^2} + {\|{\vec{n}_{\tau}}\|^2} - 2\,{\vec{n}_0}\cdot{\vec{n}_{\tau}}} ~.
\end{equation}
Because $\|{\vec{n}_t}\|\leq 1$ for all $t\in[0,\tau]$, one gets ${d_{\text{euc}}}({\vec{n}_0},{\vec{n}_{\tau}}) \leq \sqrt{2} \sqrt{1 - {\vec{n}_{\tau}}\cdot{\vec{n}_0}}$. In particular, by considering the case ${\vec{n}_{\tau}}\cdot{\vec{n}_0} \leq 0$, one gets that $ 1 - {\vec{n}_{\tau}}\cdot{\vec{n}_0} \geq 1$. In this setting, by taking into account that $\sqrt{x}\leq x$ for every $x \geq 1$, we find that $\sqrt{1 - {\vec{n}_{\tau}}\cdot{\vec{n}_0} }\leq 1 - {\vec{n}_{\tau}}\cdot{\vec{n}_0}$. Hence, for ${\vec{n}_{\tau}}\cdot{\vec{n}_0} \leq 0$, Eq.~\eqref{eq:SM0000000037} can be written as follows
\begin{equation}
\label{eq:SM0000000038}
d_{\text{euc}}({\vec{n}_0},{\vec{n}_{\tau}}) \leq \sqrt{2} \, (1 - {\vec{n}_{\tau}}\cdot{\vec{n}_0}) ~.
\end{equation}
Finally, by combining Eqs.~\eqref{eq:SM0000000035},~\eqref{eq:SM0000000036}, and~\eqref{eq:SM0000000038}, we obtain the result
\begin{equation}
\label{eq:SM0000000039}
{\tau_{QSL}^{\alpha}} \leq 2^{\frac{1}{2} + \frac{1}{\alpha}} {\tau_{QSL}^{\alpha,\beta}} ~,
\end{equation}
which holds whenever ${\vec{n}_{\tau}}\cdot{\vec{n}_0} \leq 0$. This proves Lemma~\ref{lem:LEMMA000003}.

%%%%%%%%%%%%%%%%%%%%%%%%%%%%%%%%%
%%%%%%%%%%%%%%%%%%%%%%%%%%%%%%%%%
%%%%%%%%%%%%%%%%%%%%%%%%%%%%%%%%%
%%%%%%%%%%%%%%%%%%%%%%%%%%%%%%%%%

\section{Proof of Lemma~\ref{lem:LEMMA000004}}
\label{app:00000000G}

Considering a two-dimensional system with a initial pure state, we showed that $\tau_{QSL}^\alpha \geq \tau_{QSL}^{DL}$. Moreover, in this same setting we also showed that the DL-QSL is only a particular case of the $(\alpha,\beta)$-QSL, namely that $\tau_{QSL}^{\alpha,\beta} = 2^{-1/\alpha}\tau_{QSL}^{DL}$, or equivalently $ \tau_{QSL}^{DL} = 2^{1/\alpha}\tau_{QSL}^{\alpha,\beta}$, where $\alpha,\beta \in [1,\infty]$ and $1/\alpha + 1/\beta = 1$.

Since $2^{1/\alpha}\geq 1$ for every $\alpha \in [1,\infty]$, we have 
\begin{equation}
    \tau_{QSL}^{\alpha}\geq \tau_{QSL}^{DL} = 2^{1/\alpha}\tau_{QSL}^{\alpha,\beta} \geq \tau_{QSL}^{\alpha,\beta}~,
\end{equation}
and consequently we can conclude that $\tau_{QSL}^{\alpha} \geq \tau_{QSL}^{\alpha,\beta}$. 

%%%%%%%%%%%%%%%%%%%%%%%%%%%%%%%%%
%%%%%%%%%%%%%%%%%%%%%%%%%%%%%%%%%
%%%%%%%%%%%%%%%%%%%%%%%%%%%%%%%%%
%%%%%%%%%%%%%%%%%%%%%%%%%%%%%%%%%

%%%%%%%%%%%%%%%%%%%%%%%%%%%%%%%%%
%%%%%%%%%%%%%%%%%%%%%%%%%%%%%%%%%
%%%%%%%%%%%%%%%%%%%%%%%%%%%%%%%%%
%%%%%%%%%%%%%%%%%%%%%%%%%%%%%%%%%

%\bibliography{qsl_refs_09_12_2024}


\end{document}